\documentclass[jmp]{revtex4-1}

\usepackage{enumitem}
\usepackage{amsmath}
\usepackage{amssymb}
\usepackage{amsthm,bbm}
\usepackage{mathrsfs}
\usepackage[hyperfootnotes=false]{hyperref}
\usepackage[dotinlabels]{titletoc}
\usepackage{subfigure}
\newlength\Li \newlength\Lii 
\newtheorem{thm}{Theorem}[section]
\newtheorem{prop}[thm]{Proposition}
\newtheorem{cor}[thm]{Corollary}
\newtheorem{lemma}[thm]{Lemma}
\newtheorem*{lemmanonum}{Lemma}
\newtheorem{defn}[thm]{Definition}
\newtheorem*{remarknonum}{Remark}
\usepackage{graphicx}
\newcommand{\ket}[1]{\ensuremath{\left|#1\right\rangle}}
\newcommand*{\rom}[1]{\expandafter\@slowromancap\romannumeral #1@}

\begin{document}

\title[When is a  QCA a QLGA?] {When is a  Quantum Cellular Automaton (QCA) a Quantum Lattice Gas Automaton (QLGA)?}
\author{Asif Shakeel and  Peter J. Love}
\address{Department of Physics, Haverford College, Haverford, PA 19041, USA}
\email{ashakeel@haverford.edu, plove@haverford.edu}

\begin{abstract}
Quantum cellular automata (QCA) are models of quantum computation of particular interest from the point of view of quantum simulation. Quantum lattice gas automata (QLGA - equivalently partitioned quantum cellular automata) represent an interesting subclass of QCA. QLGA have been more deeply analyzed than QCA, whereas general QCA are likely to capture a wider range of quantum behavior. Discriminating between QLGA and QCA is therefore an important question. In spite of much prior work, classifying which QCA are QLGA has remained an open problem.  In the present paper we establish necessary and sufficient conditions for unbounded, finite Quantum Cellular Automata (QCA) (finitely many active cells in a  quiescent background) to be Quantum Lattice Gas Automata (QLGA). We define a local condition that classifies those QCA that are QLGA, and we show that there are QCA that are not QLGA. We use a number of tools from functional analysis  of   separable Hilbert spaces and representation theory  of associative algebras that enable us to treat QCA on finite but unbounded configurations in full detail.
\end{abstract}

\maketitle

\section{Introduction} \label{section:intro}

Feynman first noted that simulating the full time evolution of quantum systems on classical computers is a hard problem, and that  one might use one quantum system to efficiently emulate another~\cite{bib:feynman1,bib:feynman2,bib:feynman3}. Feynman's suggestion became a founding idea in the field of quantum computation~\cite{lloyd1,bib:wiesner,bib:zalka,bib:gubernatis2,bib:gubernatis1,bib:somma} and has been subsequently developed in the field of quantum simulation, which has attracted considerable theoretical and experimental attention in recent years with applications in physics, chemistry and biology~\cite{qs:science,qs:pnas,bib:white,qs:nppqs,lhnm:udgsti,l:btf,ycmmla:fttcqc,mrekta:qsoqs}.

Simulation of quantum systems on current classical computer hardware is a well established field with simulations of diverse systems from quantum chemistry to the structure of the proton. For large systems these simulations rely on approximate methods, such as semiclassical treatments or the Monte Carlo method, which scale only polynomially in the problem size. The results of even approximate methods are of interest as they address issues which are not accessible in any other way.  Classical simulation of quantum systems will remain a hard problem for decades to come, and one may expect useful quantum computers to appear on this timescale. 

Amongst both classical and quantum simulation methods, cellular automata, lattice gas and random walk methods can be singled out for their simplicity. In 1948 von Neumann set out to show that complex phenomena can arise out of many simple, identical interacting entities. Following a suggestion by Ulam he adopted an approach in which space, time and the dynamical variables are all discrete. The result was the {\it cellular automaton}, a homogeneous array of cells with a finite number of states evolving in discrete time according to a uniform local transition rule~\cite{bib:vonneu1}. 

In the context of physical simulation, classical cellular automata led to lattice-gas models - also called partitioned cellular automata - which can simulate diffusive processes and fluid mechanics for both simple and complex fluids~\cite{bib:rz}. Lattice gases are the only broad class of cellular automata models that have enjoyed wide success in quantitative modeling of physical phenomena. 

In the quantum setting, quantum cellular automata~\cite{a:watrous, wvd:uqca,dls:dpwfqca,ds:dpulqca,a:aculqca,sw:rvqca,anw:odqca}, quantum lattice gases~\cite{bib:meyer1,bib:meyer2,bib:meyer3,bib:meyer4,bib:meyer5,Bog1,Bog2,Bog3}, quantum lattice Boltzmann methods~\cite{sb:lbeqm,s:nsse,sp:nvbs,sp:albs,ld:ctdqlb,d:eecqlb}, and  quantum random walks~\cite{bib:qrw1,bib:qrw2,bib:qrwr}, have all attracted considerable attention. In particular the one-dimensional cellular automata for the Dirac equation, originally described by Feynman in a problem in~\cite{bib:feynmanhibbs}, has been investigated and generalized by Meyer into the quantum lattice gas model~\cite{bib:meyer1,bib:meyer2,bib:meyer3,bib:meyer4,bib:meyer5}. Similar models were investigated independently by Boghosian and Taylor~\cite{Bog1,Bog2,Bog3}, and by Succi and Benzi~\cite{sb:lbeqm,s:nsse}.  Meyer also showed that the quantum lattice Boltzmann model in 1D~\cite{s:nsse} and quantum lattice gas for a single particle are equivalent~\cite{bib:meyer2}. Recently, optical implementations of quantum random walks have demonstrated topologically protected bound states~\cite{bib:topoqrw}, and simulations of the Dirac equation have been performed in trapped ions~\cite{bib:diracqrw}. 

Quantum random walks~\cite{bib:qrw1,bib:qrw2,bib:qrwr} and single-particle quantum lattice-gas models~\cite{Bog1,Bog2,Bog3,bib:meyer2} and lattice-Boltzmann methods~\cite{sb:lbeqm,s:nsse} all describe a single quantum particle moving on a lattice. Classical lattice-gas models typically describe many interacting particles moving on a lattice, and Meyer generalized one-particle models to a multi-particle quantum lattice-gas model in~\cite{bib:meyer2}. Whereas the $1+1$ dimensional single-particle quantum lattice-gas simulates the Dirac equation in $1+1$ dimensions, the multiparticle quantum lattice-gas, that allows two-particle interactions at a lattice site, can describe the massive Thirring model~\cite{bib:meyer2,t:asrft}. Because the massive Thirring model is a model of relativistic fermions with self-interactions this supports the idea that multiparticle quantum lattice-gases should be considered as models for multiparticle, interacting quantum mechanics and quantum field theories. We shall take quantum lattice-gas automata (QLGA) to mean the multiparticle models throughout.   

Quantum cellular automata, which include QLGA as a subclass, are the most general discrete, translation invariant quantum models, and one would therefore expect them to describe a broad range of phenomena arising in interacting quantum systems. However, quantum cellular automata have not attracted the same degree of interest as their classical counterparts, and have attracted relatively little attention compared to other approaches to quantum simulation. As we discuss in detail below, the definition of QCA models has been the source of some debate in the literature, which has perhaps delayed applications of these techniques. We hope the present work addresses this by clearly delineating the QLGA models from their complement in the set of QCA and providing a self-contained introduction to these models.

\subsection{Classical lattice-gas cellular automata}

Perhaps the most widely explored application of cellular automata has been in the field of fluid dynamics. Classical lattice-gases have been used extensively for modeling hydrodynamics since Frisch, Hasslacher, and Pomeau~\cite{bib:fchc}, and Wolfram~\cite{bib:w4} showed that it is possible to simulate the incompressible Navier-Stokes equations using discrete boolean elements on a lattice.  The update rule for all lattice-gases takes place in two steps, propagation and collision. For the simplest such models one represents the presence or absence of a particle by a bit. Each bit at a lattice site is associated with a link in a lattice and a corresponding velocity. During propagation the bits move in the direction of their velocity vectors along the links, retaining their velocities as they do so. During the collision step, the newly propagated bits are modified according to a purely local rule. Specification of collision rules completes the description of a lattice-gas automata (LGA). If the collision process is stochastic, a single particle in such a model will follow a random walk, and so one may regard classical lattice-gas models as equivalent to classical multiparticle random walks with an exclusion principle. 

\begin{figure}
 \includegraphics{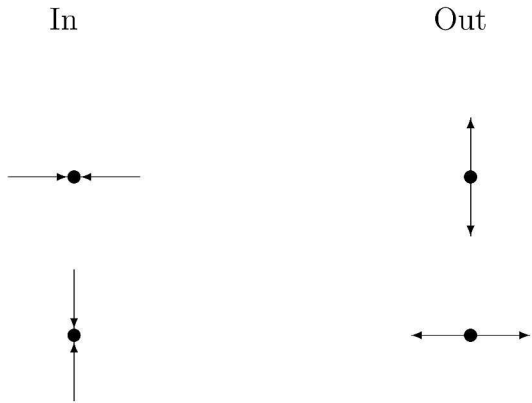}
  \caption{Non-trivial collisions of the HPP lattice-gas: The dynamics takes place on a cartesian lattice in two dimensions. A pair of incoming particles with opposite velocities in the North-South direction become a pair of particles with opposite velocities in the East-West direction, and vice versa.}
  \label{fig:hpp}
\end{figure}

The simplest two-dimensional example of a lattice-gas model is the HPP model~\cite{HPP}. Here, the underlying lattice is cartesian, hence there are four lattice vectors per site, corresponding to a particle moving north, east, south, or west. There may be at most one particle per direction per site, leading to at most four particles at each lattice site. The non-trivial collision of the HPP gas is shown in Figure~\ref{fig:hpp}. 

The distinction between cellular automata and lattice-gases has been the source of some debate~\cite{Toffoli,Henon}. A homogeneous cellular automata utilizes the same neighborhood and the same update rule at every site and every timestep, as in the Game of Life. The lattice-gas evolves in two distinct substeps: {\em propagation} and {\em collision} (described in detail below). The connection between the models may be made by introducing the concept of a block- or partitioned- cellular automata, which alternates between two update rules and two neighborhoods. The simplest of these automata employ a Margolus neighborhood~\cite{margolus1}. As an example, we shall construct the two rules and two neighborhoods which, when alternately applied to a cellular automata, yield the HPP lattice-gas dynamics described above. The update rules are shown in Figure~\ref{fig:hppcellr}. The update rules act on a block of four cells, which we shall refer to as upper left (UL), upper right (UR), lower left (LL), lower right (LR). The propagation rule is applied identically to all states: the occupation of each cell in the four cell block is exchanged with the occupation of the cell diagonally opposite. The collision step is identical to the propagation step except when UR and LL are the only cells occupied or when LR and UL are the only cells occupied, in which case the occupancies of these pairs of cells are exchanged. Graphically this corresponds to either diagonal of the four cells being occupied. In these cases the state is flipped to the other diagonal occupation, as shown in Figure~\ref{fig:hppcellr}. 
 \begin{figure}[ht]
 \includegraphics{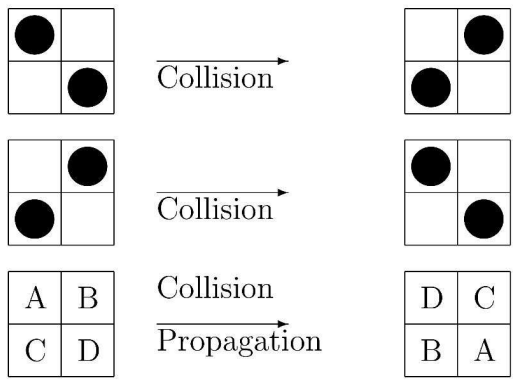}
 \caption{\label{fig:hppcellr}Update rules for the HPP partitioned cellular automata. The top two boxes show the collision rules in which the configurations of the cells which are modified in a non-trivial way. The bottom boxes show the update rule which implements propagation in the automata. The bottom rule is applied to both neighborhoods of the automata. $A,B,C,D$ represent any occupation state of the four cells, and the propagation update is the permutation of these occupations shown.}
  \end{figure}

Clearly, if either of these rules was applied to a constant neighborhood the resulting dynamics would be trivial. Neither rule can change the occupation outside the four cell block and both rules are self inverse, so applying them twice to the same block returns the original configuration. In our partitioned automata we specify a different neighborhood for each rule in order to obtain the HPP dynamics. The Margolus neighborhoods, the cells of the automata and their relationship to the original cartesian lattice on which we defined HPP model are shown in Figure~\ref{fig:hppcelln}. The two neighborhoods are the two ways of partitioning a cartesian grid into square tiles of four cells. The original cartesian lattice on which the HPP lattice gas was defined is the lattice of diagonals with a lattice node at the center of every collision neighborhood. The lattice-gas vertices are shown by open circles in the third Figure in Figure~\ref{fig:hppcelln}. 

\begin{figure}[htbp]
\centering
\subfigure[Collision neighborhood]{%
\includegraphics[width=0.3\textwidth]{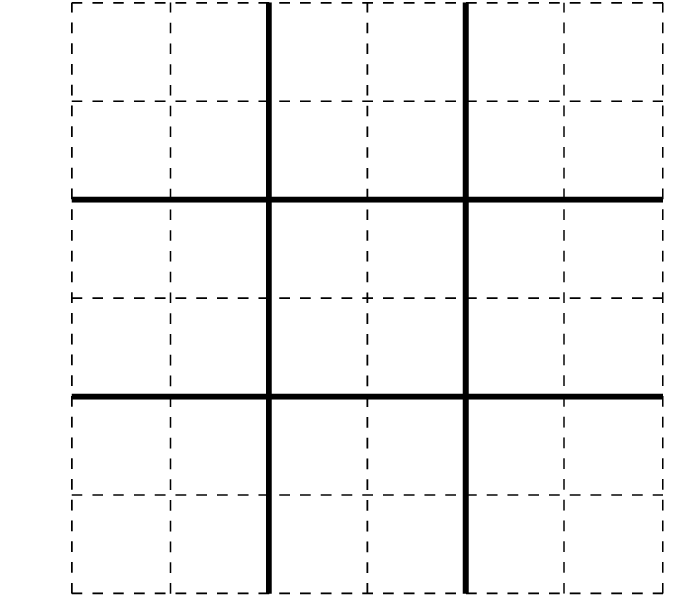}
\label{fig:subfigure1}}
\quad
\subfigure[Propagation neighborhood]{%
\includegraphics[width=0.3\textwidth]{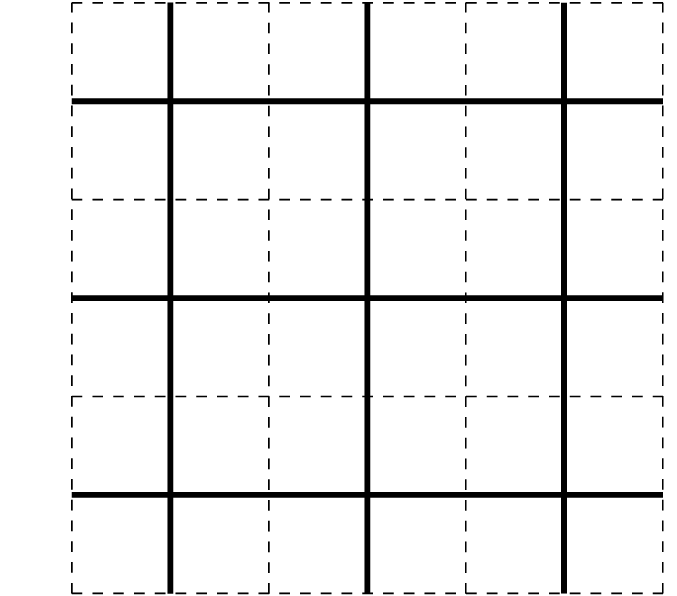}
\label{fig:subfigure2}}
\quad
\subfigure[Lattice-gas sites and particle directions]{%
\includegraphics[width=0.3\textwidth]{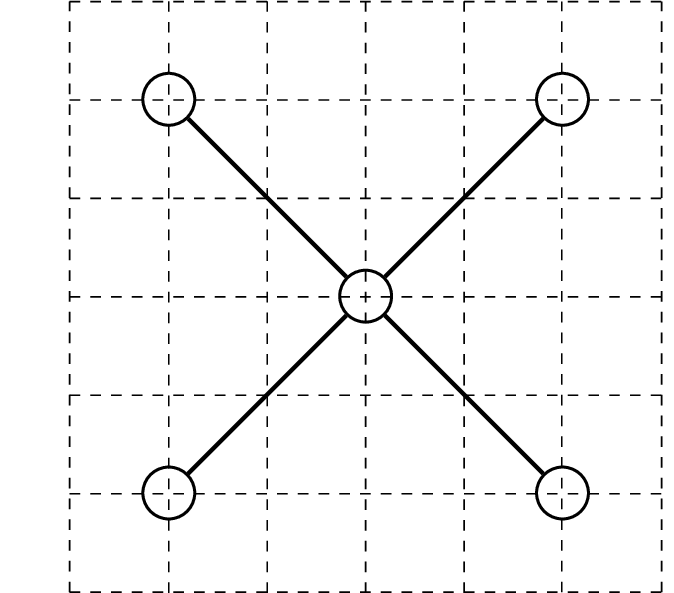}
\label{fig:subfigure3}}

\caption{\label{fig:hppcelln}Neighborhoods for collision (a) and propagation (b). The boundaries of the four cell blocks on which the rules shown in Figure~\ref{fig:hppcellr} are applied are shown by the bold lines. Figure (a) shows the boundaries for the collision neighborhood. Figure (b) shows the boundaries for the propagation neighborhood. Figure (c) shows the correspondence between the cells of the partitioned automata and the lattice of the lattice-gas. Bold lines denote the directions along which lattice-gas particles move. Open circles denote sites of the lattice-gas.}
\label{fig:figure}
\end{figure}

In this approach, the original cells of the automata are grouped together into blocks of four to make the sites of the HPP lattice-gas model. The partitioned dynamics is reinterpreted - one step of the dynamics (which acts on one partitioning of the original cells) now acts locally on a single lattice-gas site. This local action is termed the collision or scattering step. The second step of the dynamics  (which acts on the other partitioning of the original cells in the Margolus scheme~\cite{margolus1}) implements propagation of information between the lattice-gas sites. In the case of the HPP model, where the alphabet of the original cells is $\{0,1\}$, this step has the interpretation of particles moving from one lattice-gas site to another. The four sub-cells of a lattice-gas site are then reinterpreted as lattice vectors, and in the propagation step lattice-gas particles move in the direction of their lattice vector to a new site. Hence we have taken a partitioned cellular automata and reverse-engineered a lattice-gas model from it. 

The general question of when a classical cellular automata can be reinterpreted in this way was recently taken up in~\cite{Toffoli}. Because we know, from the  example of~\cite{anw:odqca}, that there are QCA that are not partitioned CA, we shall proceed by first defining classical lattice-gas automata and proving a general condition for when a QCA is a QLGA. In the process we shall extend our definition of a classical lattice-gas to the quantum case.

Consider the lattice to be $\mathbb{Z}^n$. Each point $x$ on this lattice is a site of the lattice-gas and has the same neighborhood $\mathcal{N}_x$. We may consider each neighborhood $\mathcal{N}_x$ to be a translation (by $x$) of the neighborhood $\mathcal{N}$ of the site $x=0$, where $\mathcal{N}$ is a finite subset of $\mathbb{Z}^n$. Hence the neighborhood of site $x$ is given by  ${x+z|z\in\mathcal{N}}$. The state of a single site may be constructed by assigning a substate to each of a set of lattice vectors - each vector corresponding to a neighborhood element $z\in\mathcal{N}$. The sub-state ${\bf k}_{x}(z)$ of each lattice vector $z$ at lattice site $x$ takes values in an alphabet $\mathcal{C}_z$. For the simplest lattice-gases, each lattice vector may have either one or zero particles present and so $\mathcal{C}_z =\{0,1\}~\forall~z$. However, more sophisticated models allow multiple types of particles, and three-dimensional models allow some vectors to carry more than one particle~\cite{bib:bcp}, hence in general $\mathcal{C}_z$ may contain arbitrarily many symbols and may vary with $z$. The state ${\bf k}_x$ of a lattice-gas site $x$ is therefore an element of the Cartesian product of all of the $\mathcal{C}_z$:
\begin{equation}\label{classk}
{\bf k}_x=\prod_{z\in\mathcal{N}}{\bf k}_x(z)\in B=\prod_{z\in\mathcal{N}} \mathcal{C}_z
\end{equation}
The state of the lattice-gas is an element of the set of infinite sequences $\mathcal{S}$ over $B$. 

We may now define the two substeps of the lattice-gas dynamics. Firstly, the sub-states propagate to the lattice sites in the neighborhood:
\begin{align*} %\label{propdefn}
\sigma: \mathcal{S} &\longrightarrow\mathcal{S}\\
\prod_{x \in \mathbb{Z}^n}  \prod_{z\in\mathcal{N}} {\bf k}_{x}(z)   &\mapsto \prod_{x \in \mathbb{Z}^n}   \prod_{z\in\mathcal{N}}{\bf k}_{x+z}(z)  \nonumber
\end{align*}
Secondly, in the collision step, the state of each lattice-gas site changes locally under a map $\Theta$:
\begin{equation}
\Theta:\prod_{z\in\mathcal{N}} \mathcal{C}_z\longrightarrow\prod_{z\in\mathcal{N}} \mathcal{C}_z.
\end{equation}
With these classical definitions in hand, we now turn to the definition of quantum cellular automata.

\subsection{Watrous Quantum Cellular Automata}

QCA have been studied, in considerable depth, by several authors~\cite{a:watrous, wvd:uqca,dls:dpwfqca,ds:dpulqca,a:aculqca,sw:rvqca,anw:odqca}. The main idea is to develop quantum models that retain the key features of a classical cellular automata: translation invariance and discreteness in time, space and dynamical variables. A classical deterministic cellular automata is defined as follows:
\begin{enumerate}[label=(\roman{*})]
\item{A lattice $\mathcal{L}$, whose sites we shall label by $x$.}
\item{A set of cell values $Q$ with a distinguished quiescent state $q_0$.}
\item{A neighborhood $\mathcal{N}$. This translates to a neighborhood $\mathcal{N}_x$ for each cell $x$, consisting of a set of cells as follows: $\mathcal{N}_x = \{x+z | z\in\mathcal{N}\}$.}
\item{A configuration of the automata is a map $a:L\longrightarrow Q$. If $L$ is infinite then only a finite number of cells take values in $Q$ that are not quiescent.}
\item{A local transition rule $\delta$ which updates the value $s_x\in Q$ at $x$ according to the values $s_i\in Q$ (where $i \in \mathcal{N}_x$) of the cells in the neighborhood $\delta:Q^{|\mathcal{N}|}\mapsto Q$, subject to the condition that \begin{equation}\delta: q_0^{|\mathcal{N}|} \mapsto q_0.\end{equation}}
\end{enumerate}

Early work produced an extension of the classical cellular automata (CA) defined above based on local rules~\cite{dm:u1dqca, wvd:uqca,dls:dpwfqca,ds:dpulqca,a:aculqca}. To avoid confusion with the definition of a Quantum Cellular Automata (QCA) that we shall introduce later we follow~\cite{sw:rvqca} and call these Watrous Quantum Cellular Automata (WQCA). A WQCA is defined as follows:

\begin{enumerate}[label=(\roman{*})]
\item{A lattice ${L}$, whose sites we shall label by $x$.}
\item{A set of cell values $Q$ with a distinguished quiescent state $q_0$.}
\item{A neighborhood $\mathcal{N}$. This translates to a neighborhood $\mathcal{N}_x$ for each cell $x$, consisting of a set of cells as follows: $\mathcal{N}_x = \{x+z | z\in\mathcal{N}\}$.}
\item A state of the automaton, $\Psi$, is a superposition of classical configurations.
\item A  local evolution rule $f$ which maps an element of $Q^{|\mathcal{N}|}$ to an element of   $W=\mathbb{C}[Q]$, 
\begin{align*}
f:  Q^{|\mathcal{N}|} \longrightarrow \mathbb{C}[Q]
\end{align*} 
subject to the constraint that if the neighborhood state is $q_0^{|\mathcal{N}|}$ then $f$ assigns amplitude $1$ to the quiescent state and amplitude $0$ to all other states.
\end{enumerate}
From the above data, one can construct a  global evolution rule $R$ by jointly evolving the state of each cell $x$ in any basis state by  the local rule $f$ and extending by linearity. The global evolution $R$ is required to be:
\begin{enumerate}[label=(\roman{*})]
\item \textbf{Unitary}: $R^\dag R = \mathbb{I}$.
\item \textbf{Translation invariant}:
 Let $\tau_z$, for some $z \in \mathbb{Z}$, be the translation operator defined by its action on a classical configuration $ c  = \prod_{j \in \mathbb{Z}} c_j$ for $c_j\in Q~\forall~ j$:
\begin{equation*}
\tau_z:  \prod_{j \in \mathbb{Z}}  c_j   \mapsto  \prod_{j \in \mathbb{Z}}  c_{j+z}  
\end{equation*}
$\tau_z$ is extended  by linearity. We denote the group of translations:
 $\mathcal{T} = \{\tau_z\}_{z \in  \mathbb{Z}}$.  A linear operator $R$ is \textit{translation-invariant} if  $\tau_z   R   \tau^{-1}_z = R $ for all $z \in  \mathbb{Z}$. 
\end{enumerate}

By definition, the above model of WQCA is translation-invariant provided the rule $f$ is the same for all cells. Unitarity for this model has been extensively investigated~\cite{dm:u1dqca, wvd:uqca,dls:dpwfqca,ds:dpulqca}.  This model of QCA was found~\cite{anw:odqca}, in some instances of the finite, unbounded QCA, to allow \textit{super-luminal signaling}.   Super-luminal signaling refers to the  faster than light signaling which occurs  when the state of a cell (\textit{restricted} or \textit{reduced} state,  defined in the next section) can be affected by that of another cell which is an unbounded distance away from it. This is observed when certain classical CAs with finite neighborhood schemes are \textit{quantized} (the local rule for classical CA is taken as that for the QCA).

\subsection{Axiomatic QCA}

An axiomatic  approach has been developed to overcome  the problem of  super-luminal signaling, a non-local phenomenon, seen in the WQCA models~\cite{sw:rvqca, anw:odqca, anw:ucil, sw:litqo}. The requirement of causality is imposed in the definition of  a QCA, with roots in the ideas of~\cite{bgnp:clqo,sw:litqo}, to preclude this non-locality~\cite{cd:luqca,sw:rvqca,anw:odqca,anw:ucil}. 

The axiomatic model was introduced by Schumacher and Werner in~\cite{sw:rvqca} and further developed by Arrighi, Nesme, and Werner in~\cite{anw:odqca}. In this approach a QCA is defined to be a unitary, translation-invariant, $*$-homomorphism of a $C^*$ -algebra which satisfies a quasi-local condition by which the global homomorphism R is uniquely determined by a local homomorphism of tensor factor (one-site) algebras to the neighborhood.  The program developed in~\cite{sw:rvqca} investigates when the structure of such QCAÕs can be given by $*$-isomorphisms obeying a quasi-local condition as follows. Once the cells are grouped into super-cells, such super-cells partition the lattice. Then there exists a shifted partition with super-cells that overlap with the original cells in a Margolus partitioning scheme~\cite{margolus1}. The global evolution is given by a set of unitary operators U, V, each of which acts on a cell of the respective partition, applied in succession. If this structure holds, then imposing the requirement of causality on quantum cellular automata requires them to be quantum lattice-gas automata.

One of the conclusions of Schumacher and Werner in~\cite{sw:rvqca}, that all QCA are QLGA, was revised by Arrighi, Nesme, and Werner in~\cite{anw:odqca}. Arrighi, Nesme, and Werner~\cite{anw:odqca} give an example of a QCA in~\cite{anw:odqca} that is not a QLGA, thus showing that there is a separation between two non-empty subclasses of QCA. Precisely characterizing this separation by determining which QCA are QLGA is the subject of the present paper. 
The axiomatic one-dimensional development of Arrighi, Nesme, and Werner~\cite{anw:odqca} is more restricted and defined on a Hilbert space of unbounded, finite configurations (finitely many active cells in a quiescent background). This set of sequences, called the set of finite configurations, is embedded into some abstract Hilbert space in which the elements of this set comprise an orthonormal basis. This is a countable-dimensional Hilbert space Ð the Hilbert space of finite configurations. In this context, they define a QCA as a unitary, translation-invariant and causal operator on the Hilbert space of finite configurations.

In the present work we further develop the axiomatic approach and place it on a new mathematical foundation. We incorporate in the development of axiomatic QCAs the topology of the Hilbert space. This is an essential part of any study of the algebra of operators, particularly infinite dimensional algebras. This is intimately connected, in this case, with the idea of local algebras and the part they play in the determination of the global evolution. This  topology provides the foundation for connecting the local algebras to the algebra of bounded linear operators of the Hilbert space.  With this topology, we can investigate  the space of bounded linear operators  as a  von Neumann Algebra (as well as  a $C^*$ algebra), making  the local subalgebras a useful tool.

We develop a  model on the same space of  finite but unbounded configurations as Arrighi, Nesme, and Werner~\cite{anw:odqca}. To make a Hilbert space from it, we give  an inner product  structure to the set of finite configurations that is  a natural extension of  the inner product of the Hilbert space of each cell. We define the Hilbert space of finite configurations as the $\ell^2$ completion of the linear span of the set of finite configurations. This defines the separable topology of the Hilbert space and the induced inner product is used to define the norm, weak, and strong  topologies of the bounded linear   operators on the space. At the outset, we show how the subalgebra of local finite dimensional operators is connected to the algebra of bounded linear operators on the whole space (Theorem~\ref{Zdense}, the Density Theorem).  This provides the background for the development of some of the key ideas, and underpins the proof of the Structural Reversibility Theorem (Arrighi, Nesme, and Werner~\cite{anw:odqca}).

The main result of this paper is that QLGA can be characterized by a local condition on QCA. This condition pertains to the  set of image algebras under the global evolution of the neighboring cells.  The condition for a QCA to be a QLGA is that the local pieces of these image algebras generate the full cell algebra.  It is clear from our  development  that a central result, Theorem~\ref{thmS}, about the tensor product decomposition  of Hilbert space of  a single cell is not true in general. In particular, to have that decomposition,  the aforementioned subalgebra condition is needed. This condition is not satisfied for all QCA, as shown by the example in~\cite{anw:odqca}. This means that there is a class of QCAs that are not QLGAs and that require further exploration. 

The current paper is structured as follows. In Section~\ref{bckg} we introduce the Hilbert space of interest for our automata. In Section~\ref{axmod} we introduce the axiomatic definition of QCA and prove several properties of operators that are translation invariant or causal, or both. We extend our definitions of the propagation and collision operators to the quantum case and formally define a QLGA before proving our main theorem. In Section~\ref{ex} we give a number of examples of QCA, for the case when our condition is satisfied and when it is not. We close the paper with some conclusions and discussion. In the interests of making the present paper self contained we collect a number of background and corollary results in the appendices.
  
\section{The Hilbert space of finite configurations} \label{bckg}
{\ \\} 
We begin with a finite set of symbols $Q$ containing a special \textit{quiescent} symbol $q_0$, and an infinite lattice $\mathcal{L} = \mathbb{Z}^n$ in $n$ dimensions. Let us denote by $W$ the Hilbert space of formal linear combinations of symbols in $Q$, i.e.  $W = \mathbb{C}[Q]$ (in this paper all of the  vector spaces will be over $\mathbb{C}$). This Hilbert space $W$ is the Hilbert space of quantum states of a single cell (for  the rest of the paper, we will use the terms \textit{cell} and \textit{site} to mean an element of the lattice $\mathcal{L}$). Let the dimension of $W$ (cardinality of $Q$) be $\text{dim} W = d_W$. The set of  basis elements of $W$ corresponding to the symbols in $Q$ is denoted $B_Q$: 
\begin{equation} \label{BQ}
B_Q = \{\ket{q} : q\in Q\}
\end{equation}
$W$ has an inner product:
\begin{equation}\label{sofcip}
\langle \alpha \vert \vartheta \rangle = \sum_{q \in Q} \overline{\alpha_q}\vartheta_q
\end{equation}
where $\alpha  = \sum_{q \in Q} \alpha_q \vert q \rangle$,  $\vartheta  = \sum_{q \in Q} \vartheta_q \vert q \rangle$. 

Informally, the Hilbert space of the automata is the infinite tensor product of copies of $W$. The idea of infinite tensor products was first studied by von Neumann in his landmark paper ``On Infinite Direct Products"~\cite{vn:idp}. Later,  Guichardet~\cite{ag:shrt}  introduced a more modern way of describing an infinite tensor product of Hilbert spaces enumerated by a countable set. This is sometimes referred to as the \textit{incomplete tensor product} construction. From countably many copies of the  same Hilbert space $W$ the incomplete tensor product is obtained as the inductive limit  of an ascending chain of finite tensor products of $W$. This space has a natural basis that we call the \textit{set of finite configurations}.  The \textit{Hilbert space of finite configurations} is the incomplete tensor product construction of Guichardet with the set of finite configurations as its basis~\cite{ag:shrt}. 
\begin{defn}\label{sofc}
The \textit{set of finite configurations}, denoted by $\mathcal{C}$, is  the set of  simple tensor products with only  finitely many active elements, 
\begin{equation*} \mathcal{C} := \{ \bigotimes_{x \in \mathbb{Z}^n} \vert c_x \rangle : \vert c_x \rangle \in B_Q,  \text{ all but finite } \vert {c}_{x} \rangle  = \vert q_0 \rangle \}
\end{equation*} 
 \end{defn}
 
Thus, $\mathcal{C}$ is a countably infinite set. Let $\vert c \rangle = \bigotimes_{x \in \mathbb{Z}^n} \vert c_x \rangle, \vert c' \rangle = \bigotimes_{x \in \mathbb{Z}^n} \vert {c_x}' \rangle \in \mathcal{C}$. Define the inner product of the elements $\vert c \rangle$, $\vert c' \rangle$:
\begin{equation*}
\langle c \vert c' \rangle = \prod_{x \in \mathbb{Z}^n} \langle c_x \vert {c_x}' \rangle
\end{equation*}
and extend it by linearity to get an inner product on  span($\mathcal{C}$) (the set of finite linear combinations of elements of $\mathcal{C}$). 

\begin{defn}\label{hofc}
The \textit{Hilbert space of finite configurations},   denoted by $\mathcal{H}_\mathcal{C}$, is the ${\ell}^2$ completion of $\text{span}(\mathcal{C})$  under the norm  induced by the above inner product.  $\mathcal{C}$ constitutes an orthonormal basis of $\mathcal{H}_\mathcal{C}$.
\end{defn}
 By definition, $\mathcal{H}_\mathcal{C}$ is a separable Hilbert space since  it has a countable orthonormal  basis $\mathcal{C}$ (this follows from a  standard  theorem on  Hilbert spaces in  texts on analysis,  for example, Folland~\cite{gf:ra}, Proposition 5.29, pg. 176). 
\begin{defn}
The \textit{neighborhood} $\mathcal{N} \subset \mathbb{Z}^n$ is some finite set. The \textit{neighborhood of a cell  $x \in \mathbb{Z}^n$},  denoted by $\mathcal{N}_x$, is a translation of $\mathcal{N} $ to $x$.
\begin{equation*} %\label{Nx}
\mathcal{N}_x = \{k+x : k \in \mathcal{N}\} 
\end{equation*}
\end{defn}
The state of the QCA, $\rho$, is described by a \textit{density operator}.
\begin{defn}
A \textit{density operator}, $\rho$,  is a positive trace class operator on  $\mathcal{H}_\mathcal{C}$, with tr($\rho$)~=~ 1.
 \end{defn}

%%%%%%%%%%%%%%%%%%%%%%%%%%%%%%%%%%%%%%%%%%%%%%%%%%%%%%%%%%%%%%%%% 
\section{Axiomatic Definition of QCA}~\label{axmod}

In this section we introduce the axiomatic definition of QCA in which the requirement of causality is built in from the start. We show that a QCA on the Hilbert space of finite configurations must possess a quiescent state as an eigenstate of its global evolution. We define local operators, and connect the algebraic structure of these operators to the topological structure of operators on the Hilbert space of finite configurations. We then consider the action of the global evolution operator on the local operators and define the requirement of causality in terms of the evolution of local operators. This definition of causality in the Heisenberg picture was termed dual causality or structural reversibility by~\cite{sw:rvqca} and ~\cite{anw:odqca}. Structural reversibility is useful because it means that the analysis of the causal global evolution can proceed by analysis of finite dimensional algebras of local operators.

We  continue from the  definitions of Section~\ref{bckg}. We need a few more definitions to properly state the requirements for an axiomatic  model of QCA.  For a  finite subset $D \in \mathbb{Z}^n$, we introduce the idea of  a  \textit{co}-$D$ space, which will be of use for a number of results. 
\begin{defn} \label{coDdef}
Let $D \in \mathbb{Z}^n$ be a  finite subset. Define the  set of co-$D$ configurations to be $\mathcal{C}_{\overline{D}} := \{ \bigotimes_{i \in \mathbb{Z}^n\setminus D} \vert {c}_i \rangle : {c}_i \in {Q},  \text{ all but finite } \vert {c}_{i} \rangle  = \vert {q}_0 \rangle \}$.  Let the inner product on $\mathcal{C}_{\overline{D}} $ be induced by the inner product on $W$~\eqref{sofcip}, as in the case of   ${\mathcal{C}}$.  Then the \textit{co}-$D$ space, denoted $\mathcal{H}_{{\mathcal{C}}_{\overline{D}}}$, is defined as  the  completion  of $\text{span}({\mathcal{C}}_{\overline{D}})$ under the induced $\ell^2$ norm.
\end{defn}

Any operator $E \in B(\mathcal{H}_{{\mathcal{C}}})$ can be expressed as a finite sum:
\begin{equation} \label{opbrk}
E = \sum_{l}  e_{l}   \otimes E_{l}
\end{equation}
where $ \{e_{l} \} \subset \text{End}(\bigotimes_{k \in D} W)$ is a linearly independent set, and $\{E_{l} \} \subset B(\mathcal{H}_{{\mathcal{C}}_{\overline{D}}})$.

We can define  \textit{partial trace} on trace class operators in general, and in particular on density operators. We let   $\{ \vert w_x \rangle \}$ be some orthonormal basis of  the co-$D$ space $\mathcal{H}_{{\mathcal{C}}_{\overline{D}}}$. Then a density operator $\rho$ can be expressed as:
\begin{equation*}
\rho = \sum_{x,y} \rho_{x,y} \otimes \vert w_x \rangle \langle w_y \vert 
\end{equation*}
where $\rho_{x,y} \in  \begin{rm}{End}\end{rm}(\bigotimes_{j \in D} W_j)$. The \textit{partial trace of $\rho$ over the tensor factors not in $D$}, denoted $\text{tr}_{\{x \notin D\}}$,   is the  sum:
\begin{equation*}
\text{tr}_{\{x \notin D\}}(\rho) := \sum_{y} \rho_{y,y}.
\end{equation*}
The sum converges, and the definition  is independent of the choice of basis $\{ \vert w_y \rangle \}$ since $\rho$ is trace class.
 
We define the \textit{restricted} or \textit{reduced} density operator.
\begin{defn}
Let $D \in \mathbb{Z}^n$ be a  finite subset. Let $\rho$ be a state (density operator) on  ${\mathcal{H}}_{\mathcal{C}}$.  The  \textit{restriction} of $\rho$ to $D$, denoted $\rho |_{D}$,  is defined to be   the partial trace of $\rho$ over  the tensor factors  not in $D$:
\begin{equation*}
\rho |_{D} := \text{tr}_{\{x \notin D\}}(\rho).
\end{equation*}
\end{defn}

Now we have the terminology to  state the requirements for  an axiomatic  QCA. 
\begin{defn}\label{qcadef}
The global evolution $R$ of a QCA  on the Hilbert space of finite configurations ${\mathcal{H}}_{\mathcal{C}}$, with neighborhood $\mathcal{N}$, is required to be:
\begin{enumerate}[label=(\roman{*})]
\item \textbf{Unitary}: $R^\dag R = \mathbb{I}$.
\item \textbf{Translation invariant}: Let $\tau_z$, for some $z \in \mathbb{Z}^n$, be the translation operator defined by its action on an element $\vert c \rangle = \bigotimes_{j \in \mathbb{Z}^n} \vert c_j\rangle \in \mathcal{C}$:
\begin{equation*}
\tau_z:  \bigotimes_{j \in \mathbb{Z}^n} \vert c_j \rangle  \mapsto  \bigotimes_{j \in \mathbb{Z}^n} \vert c_{j+z} \rangle 
\end{equation*}
The map  $\tau_z$ is extended by linearity to span($\mathcal{C}$), on which it is inner product  preserving. Then $\tau_z$ can be unitarily extended to   ${\mathcal{H}}_{\mathcal{C}}$ (that such an extension exists, and is a unitary operator on ${\mathcal{H}}_{\mathcal{C}}$, follows from the \textit{Bounded Linear Transformation} (B.L.T.) Theorem, standard in the theory of operators on Banach spaces. The reader can find it in Reed and Simon~\cite{rs:fa1} as Theorem 1.7, pg. 9). We denote the group of translations:
 $\mathcal{T} = \{\tau_z\}_{z \in  \mathbb{Z}^n}$.
 A linear operator $M$ on ${\mathcal{H}}_{\mathcal{C}}$ is \textit{translation-invariant} if  $\tau_z   M   \tau^{-1}_z = M $ for all $z \in  \mathbb{Z}^n$. 
\item \textbf{Causal relative to the neighborhood $\mathcal{N}$}: 
A linear operator $M$ on ${\mathcal{H}}_{\mathcal{C}}$ is said to be \textit{causal} relative to a neighborhood $\mathcal{N}$ if  for every pair $\rho, \rho'$,  of   density operators, and  $x \in \mathbb{Z}^n$, that satisfy:
\begin{equation*}
\rho|_{\mathcal{N}_x} = \rho'|_{\mathcal{N}_x}\text{,}
\end{equation*}
 the  operators $M \rho M^\dag, M\rho'M^\dag$  satisfy:
\begin{equation*}
M \rho M^\dag |_x = M\rho'M^\dag |_x
\end{equation*}
\end{enumerate}
\end{defn}

Let us see what the constraints of Definition~\ref{qcadef} tell us about an operator. First we consider translation-invariance. We begin by  identifying the invariants (the fixed points),  in ${\mathcal{H}}_{\mathcal{C}}$, of the group of translations $\mathcal{T} = \{ \tau_z \}_{z \in \mathbb{Z}^n}$. 
\begin{lemma}\label{tauinvar}
The   space of  $\mathcal{T}$-invariants in ${\mathcal{H}}_{\mathcal{C}}$ is one-dimensional. It is: $\mathbb{C}^\times \bigotimes_{x \in \mathbb{Z}^n} \vert q_0 \rangle$, where $\mathbb{C}^\times=\mathbb{C}\setminus \{0\}$.
\end{lemma}

\begin{proof}
Consider the action of  $\mathcal{T}$ on ${\mathcal{H}}_{\mathcal{C}}$.   Let us write  $\phi \in {\mathcal{H}}_{\mathcal{C}}$ in a basis $\{\vert C_i\rangle: i\in\mathbb{Z}\}$ of ${\mathcal{H}}_{\mathcal{C}}$ :
\begin{equation*}
\phi = \sum_{i \in \mathbb{Z}} \phi_i  \vert C_i \rangle 
\end{equation*}
where $\vert C_i \rangle \in \mathcal{C}$. Let $\vert C_0 \rangle = \bigotimes_{x \in \mathbb{Z}^n} \vert q_0 \rangle$.
If $\tau_z \in \mathcal{T}$,  then:
\begin{equation*}
\tau_z(\phi) = \tau_z(\sum_{i} \phi_i \vert C_i \rangle )  =  \sum_{i \in \mathbb{Z}} \phi_i  \tau_z(\vert C_{i} \rangle)  
\end{equation*}
Since $\tau_z$ permutes elements of $\mathcal{C}$, we can associate a permutation $s_z$ of  the indices $i \in \mathbb{Z}$ with the action of $\tau_z$, i.e., $\tau_z(\vert C_{i} \rangle) = \vert C_{s_z(i)} \rangle$.

If  $\phi$  is an invariant of $\mathcal{T}$, then:  $\tau_z(\phi) = \phi$ for all  $z \in \mathbb{Z}^n$, i.e.,
\begin{equation*}
\sum_{i \in \mathbb{Z}} \phi_i  \vert C_i \rangle   =  \sum_{i \in \mathbb{Z}} \phi_i  \tau_z(\vert C_{i} \rangle)   =  \sum_{i \in \mathbb{Z}} \phi_{{s_z}^{-1}(i)} \vert C_{i} \rangle
\end{equation*}

This implies $\phi_{i} =\phi_{{s_z}^{-1}(i)}$   for $z \in \mathbb{Z}^n$, i.e., $\phi_{i}$ are constant on the orbits of  $\mathcal{C}$ under the action of $\mathcal{T}$. But  $\mathcal{T}$ fixes  $\vert C_0 \rangle =  \bigotimes_{x \in \mathbb{Z}^n} \vert q_0 \rangle$, and all other orbits  in $\mathcal{C}$ are countably infinite in cardinality.
Since  $\phi$ has a finite $\ell^2$ norm,  this  implies  that $\phi_{i} = 0$ for all  $i \neq 0$, i.e.,   the space of $\mathcal{T}$-invariants  is  one-dimensional: $\mathbb{C} \bigotimes_{x \in \mathbb{Z}^n} \vert q_0 \rangle$. 
\end{proof}

Now that we know that the only vectors fixed by the group of translations are scalar multiples of the \textit{quiescent state}: $\bigotimes_{x \in \mathbb{Z}^n} \vert q_0 \rangle$,  we go on to  determine the action of invertible  translation-invariant operators on the quiescent state. 

\begin{lemma} \label{Rinvariant}
An invertible  and translation-invariant operator $M$ on ${\mathcal{H}}_{\mathcal{C}}$ has $\bigotimes_{x \in \mathbb{Z}^n} \vert q_0 \rangle$ as an eigenvector:
\begin{equation*}
M\big(\bigotimes_{x \in \mathbb{Z}^n} \vert q_0 \rangle \big)= \lambda_0 \bigotimes_{x \in \mathbb{Z}^n} \vert q_0 \rangle 
\end{equation*}
for some $\lambda_0 \in \mathbb{C}^\times$. In particular, if $M$ is unitary and translation-invariant, then $\lambda_0 = e^{i \Theta_0}$ for some $\Theta_0 \in \mathbb{R}$.
\end{lemma}
\begin{proof}
$M$ is  invertible, so  $M \big(\bigotimes_{x \in \mathbb{Z}^n} \vert q_0 \rangle \big) \neq 0$. And $M$ is  translation invariant,  hence:
\begin{equation*}
\tau_z   M   \tau_z^{-1} \big(\bigotimes_{x \in \mathbb{Z}^n} \vert q_0 \rangle \big) = M \big(\bigotimes_{x \in \mathbb{Z}^n} \vert q_0 \rangle \big)  \quad \forall z \in  \mathbb{Z}^n
\end{equation*}
But  $\tau_z^{-1}  (\bigotimes_{x \in \mathbb{Z}^n} \vert q_0 \rangle) = \bigotimes_{x \in \mathbb{Z}^n} \vert q_0 \rangle$ for all  $z \in  \mathbb{Z}^n$. This implies:
\begin{equation*}
\tau_z \bigg( M  \big(\bigotimes_{x \in \mathbb{Z}^n} \vert q_0 \rangle \big) \bigg) = M \big(\bigotimes_{x \in \mathbb{Z}^n} \vert q_0 \rangle \big)  \quad \forall z \in  \mathbb{Z}^n
\end{equation*}
that is, $M \big(\bigotimes_{x \in \mathbb{Z}^n} \vert q_0 \rangle \big)$ is an invariant of $\tau_z$ for all  $z \in  \mathbb{Z}^n$. Lemma~\ref{tauinvar} now implies that:
\begin{equation*}
M\big(\bigotimes_{x \in \mathbb{Z}^n} \vert q_0 \rangle \big)= \lambda_0 \bigotimes_{x \in \mathbb{Z}^n} \vert q_0 \rangle 
\end{equation*}
for some $\lambda_0 \in \mathbb{C}^\times$.

As a special case, if $M$ is unitary then all its eigenvalues are roots of unity and so $\lambda_0=e^{i\Theta_0}$ for some $\Theta_0 \in \mathbb{R}$.
\end{proof}

Finally, we consider the implication of causality on a unitary operator. First we need the related  notion of a  \textit{local} operator.  For any finite-dimensional vector spaces $V$ and $W$, we will assume the standard isomorphisms: ${\rm End}(V) \cong V \otimes  V^*$ and   ${\rm End}(V\otimes W) \cong {\rm End}(V) \otimes {\rm End}(W)$, and make use of them as needed, without explicit mention. Consider the embedding of  $\bigotimes_{j \in D} \begin{rm}{End}\end{rm}(W)$ into a subalgebra of $B({\mathcal{H}}_{\mathcal{C}})$   (the algebra  of bounded linear  operators on $\mathcal{H}_{\mathcal{C}}$, where the norm is the usual operator norm):
\begin{align} \label{algembd}
  \iota_D:     \bigotimes_{j \in D} \begin{rm}{End}\end{rm}(W)  &\hookrightarrow B({\mathcal{H}}_{\mathcal{C}}) \\ 
  a &\mapsto a \otimes \mathbb{I}_{\overline{\mathcal{D}}} \nonumber 
\end{align}
where $a$ is an element of $\bigotimes_{j \in D} \begin{rm}{End}\end{rm}(W)$, and $\mathbb{I}_{\overline{\mathcal{D}}}$ is the identity operator on the co-$D$ space, $\mathcal{H}_{{\mathcal{C}}_{\overline{D}}}$, with the operator decomposition as in~\eqref{opbrk}. Through the embedding $\iota_D$~\eqref{algembd} the algebra $\bigotimes_{j \in D} \begin{rm}{End}\end{rm}(W)$ is  isomorphic to the corresponding finite dimensional  subalgebra of $B({\mathcal{H}}_{\mathcal{C}})$.  

\begin{defn}
An operator $M$ on ${\mathcal{H}}_{\mathcal{C}}$ is \textit{local} upon a finite subset $D \in \mathbb{Z}^n$ if it is in the  image of the  map $\iota_D$~\eqref{algembd}.
\end{defn}

To understand the significance of these local subalgebras, we connect the  algebraic structure  of the subalgebra they generate to the topological structure of $B({\mathcal{H}}_{\mathcal{C}})$.  Let us denote the subalgebra embeddings~\eqref{algembd} by $\mathcal{A}_D = \iota_D(\bigotimes_{j \in D} \begin{rm}{End}\end{rm}(W))$.  Let $\{D_k \subset \mathbb{Z}^n : |D_k| < \infty\}_{k\in \mathbb{N}}$ be a strictly ascending chain:
\begin{equation*}
D_0 \subsetneq D_1 \subsetneq \ldots 
\end{equation*}
such that: $\mathbb{Z}^n = \cup_{k \in \mathbb{N}} D_k$.    We have an  ascending chain of subalgebras formed by embeddings $\{\mathcal{A}_{D_k}\}_{k\in \mathbb{N}}$, and we denote the union of these subalgebras by $\mathcal{Z}$:
\begin{equation} \label{Zdef}
\mathcal{Z} = \cup_{k} \mathcal{A}_{D_k}
\end{equation} 
For the proof of the next theorem, we need the following definition. Define, for any subset $\mathcal{S} \subset B({\mathcal{H}}_{\mathcal{C}})$,  the \textit{commutant} of  $\mathcal{S}$:
\begin{equation*} 
\text{Comm}(\mathcal{S}) = \{ y \in B({\mathcal{H}}_{\mathcal{C}}): ys = sy \:\: \forall s \in \mathcal{S} \}
\end{equation*}
 \begin{thm}[Density Theorem] \label{Zdense}
 $\mathcal{Z}$ is dense in $B({\mathcal{H}}_{\mathcal{C}})$ in the weak and strong operator topologies.
\end{thm}
\begin{proof}
 $\mathcal{Z}$ acts irreducibly on $\mathcal{H}_{\mathcal{C}}$. Indeed,    $\mathcal{Z}$ acts transitively on a dense subspace, $\text{span}(\mathcal{C})$, of $\mathcal{H}_{\mathcal{C}}$.  This implies, by Schur's Lemma (Lemma~\ref{srlem}),  Comm$(\mathcal{Z}) = \mathbb{C} \mathbb{I}$ , where $\mathbb{I}$ is the identity operator on $\mathcal{H}_{\mathcal{C}}$.  Then Comm(Comm($\mathcal{Z})) = B({\mathcal{H}}_{\mathcal{C}})$. By  von Neumann Density theorem (Theorem~\ref{vndnsthm}), Comm(Comm($\mathcal{Z}$)) is the weak and strong closure of $\mathcal{Z}$. 
 \end{proof}

The embedding $\iota_D$~\eqref{algembd} allows us to consider the individual elements of  $\mathcal{A}_x$ (the special case when $D=\{x\}$, i.e., a single cell) as acting only on the tensor factors of interest. Then we are justified, for notational reasons, in replacing the finite dimensional algebras $\mathcal{A}_D$ with  a finite tensor product of cell algebras  of the form $\mathcal{A}_x$. 

For  a unitary operator $M$  on $B(\mathcal{H}_\mathcal{C})$, denote by $C_M$, the conjugation by $M$ map on  $B(\mathcal{H}_{\mathcal{C}})$:
\begin{align*} 
C_M :  B(\mathcal{H}_\mathcal{C}) &\longrightarrow B(\mathcal{H}_\mathcal{C}) \\
						Z &\mapsto M^\dag Z M \nonumber 
\end{align*}  
As $M$ is assumed to be unitary, this map is  a unitary isomorphism (under the usual operator norm) of $B({\mathcal{H}}_{\mathcal{C}})$.  

For any set $Y\subset B(\mathcal{H}_\mathcal{C})$ the  image of $Y$ under conjugation by $M$ is $C_M(Y)$. By abuse of notation, we say,  $C_M(Y) = M^\dag Y M$. The images of the cell algebras under conjugation by $M$, $\mathcal{M}_z = M^\dag \mathcal{A}_z M$, $z \in  \mathbb{Z}^n$ are finite dimensional, and  for $z,z' \in \mathbb{Z}^n$, $z\neq z'$,  $\mathcal{M}_z$, $\mathcal{M}_{z'}$ pairwise commute as $\mathcal{A}_z$, $\mathcal{A}_{z'}$ pairwise commute.

Having defined the local operators we now consider the supports of some such operators of particular interest. Let the \textit{reflected} neighborhood  of $\mathcal{N}$, denoted $\mathcal{V}$, be:
\begin{align*} 
\mathcal{V} = -\mathcal{N} = \{-k : k \in \mathcal{N}\}
\end{align*} 
Then this reflected neighborhood $\mathcal{V}$ can also be translated. We denote the translation by $\mathcal{V}_x$.  
\begin{align*} 
\mathcal{V}_x = \{x-k : k \in \mathcal{N}\}
\end{align*} 
The  size of this set, by definition, is  $|\mathcal{V}_x | = |\mathcal{N}|$.  

It is straightforward to see, by symmetry,  that $\mathcal{V}_x \subset  \mathbb{Z}^n $ is  the subset of the lattice consisting of those elements whose neighborhood contains $x$:
\begin{align*}
\mathcal{V}_x = \{l  : x \in \mathcal{N}_l\}
\end{align*}

Causality can be expressed in the Heisenberg picture in which one considers the evolution of operators. This form of causality is captured by the  \textit{Structural Reversibility} theorem due to Arrighi, Nesme, and Werner~\cite{anw:odqca}. In the interests of making the present paper self contained we have included the proof of this Theorem in the Appendix (Theorem~\ref{strucrev}).

\begin{thm}[Structural Reversibility]  \label{strucreva}
  Let $M : \mathcal{H}_{\mathcal{C}} \longrightarrow \mathcal{H}_{\mathcal{C}}$ be a unitary  operator  and $\mathcal{N}$ a neighborhood. Let $\mathcal{V} = \{-k | k\in \mathcal{N}\}$. Then the following are equivalent:
 \begin{enumerate}[label=(\roman{*})] 
\item \label{strv1} $M$ is causal relative to the  neighborhood $\mathcal{N}$. 
\item \label{strv2}  For every  operator $A_x$ local upon cell $x$, $M^\dag A_x M$ is local upon  $\mathcal{N}_x$.
\item \label{strv3} $M^\dag$ is causal relative to the  neighborhood $\mathcal{V}$. 
\item \label{strv4}  For every  operator $A_x$ local upon cell $x$, $M A_x M^\dag$ is local upon  ${\mathcal{V}}_x$.
\end{enumerate}   
\end{thm}
 
The requirement of causal evolution will be useful in the local analysis of the global evolution $R$. This is because, by Theorem~\ref{Zdense} the subalgebra generated by the local algebra is dense in $B({\mathcal{H}}_{\mathcal{C}})$. The unitary, causal evolution of  the QCA can therefore be considered in terms of finite dimensional local subalgebras of $B({\mathcal{H}}_{\mathcal{C}})$. 

For a site $x$ the set of sites whose neighborhood contains $x$ is $\mathcal{V}_x$. We first show that the algebra $\mathcal{A}_x$ is a subset of the algebra formed from the span of the images of the local algebras $\mathcal{A}_k$ for $k\in \mathcal{V}_x$. This result is true for all QCAs.

\begin{thm}\label{thmsubalg}
Let  $M$ be a unitary and causal map on $\mathcal{H}_\mathcal{C}$ relative to some neighborhood $\mathcal{N}$.  Then for every $x\in \mathbb{Z}^n$,  
\begin{equation*}
\mathcal{A}_x \subset  \begin{rm}{span}\end{rm}(\prod_{k \in \mathcal{N}} \mathcal{M}_{x-k})
 \end{equation*}
In particular $ \mathcal{Z} = M^\dag\mathcal{Z}M$.

\end{thm}
\begin{proof}
 By  Structural Reversibility, Theorem~\ref{strucreva}~\ref{strv3}, for every $x\in \mathbb{Z}^n$,   \begin{equation*} M \mathcal{A}_x M^\dag \subset \bigotimes_{k \in \mathcal{N}}  \mathcal{A}_{x-k} \end{equation*} This implies:
 \begin{equation*} \mathcal{A}_x  = M^\dag M \mathcal{A}_x M^\dag M \subset  M^\dag (\bigotimes_{k \in \mathcal{N}}  \mathcal{A}_{x-k}) M= \begin{rm}{span}\end{rm}(\prod_{k \in \mathcal{N}} M^\dag \mathcal{A}_{x-k} M) = \begin{rm}{span}\end{rm}(\prod_{k \in \mathcal{N}} \mathcal{M}_{x-k})\end{equation*}
  The last statement follows  by the first statement and  Theorem~\ref{strucreva}~\ref{strv2}, i.e.,  \begin{equation*} \mathcal{M}_z \subset \bigotimes_{k \in \mathcal{N}}  \mathcal{A}_{z+k} \text{ for all }z \in \mathbb{Z}^n \end{equation*}  
 \end{proof}

Consider conjugation by $R$ on $B(\mathcal{H}_{\mathcal{C}})$, where $R$ is the global evolution of a QCA. We denote this map $C_R$: 
\begin{align} \label{conjR}
C_R :  B(\mathcal{H}_\mathcal{C}) &\longrightarrow B(\mathcal{H}_\mathcal{C}) \\
						Z &\mapsto R^\dag Z R \nonumber 
\end{align}  
Denote $\mathcal{R}_z = R^\dag \mathcal{A}_z R$, the image of cell algebras $\mathcal{A}_z$ under conjugation by $R$~\eqref{conjR}. This is the algebra of operators which are the images of operators localized on a single cell, after a single timestep of evolution by $R$. We are interested in the intersection of these time evolved algebras with the algebras localized on a single cell. Let us denote by $\mathcal{D}_{z,x}$ the  following subalgebra of $\mathcal{R}_z$, $z \in \mathbb{Z}^n$: 
\begin{equation} \label{Projectiondef}
\mathcal{D}_{z,x} = \mathcal{R}_z   \cap \mathcal{A}_{x}
\end{equation}

When $z \in \mathcal{V}_x$ then $\mathcal{D}_{z,x} $  are the elements of $\mathcal{R}_z$ which are contained in $\mathcal{A}_x$, where
\begin{align*} 
\mathcal{R}_z =R^\dag \mathcal{A}_zR \subset  \bigotimes_{ k \in \mathcal{N}_z} \begin{rm}{End}\end{rm}(W)  
\end{align*}
and:
\begin{align*}
\mathcal{A}_z =  \underbrace{\begin{rm}{End}\end{rm}(W)}_{k = z}   \otimes \bigotimes_{k \in \mathcal{N}_z\setminus \{z\}} \mathbb{I}_{k}.
\end{align*}

The subalgebras $\mathcal{D}_{z,x}$ for $z\in \mathcal{V}_x$ may be understood intuitively as the subalgebras localized on $x$ that the cell $x$ receives from cells in $\mathcal{V}_x$. It turns out that the algebra generated by $\mathcal{D}_{z,x}$ for $z\in \mathcal{V}_x$ is the key ingredient in understanding the structure of a subset of QCA. In particular we are interested in the case when the algebra generated by $\mathcal{D}_{z,x}$ for $z\in \mathcal{V}_x$ is equal to the full cell algebra $\mathcal{A}_x$. We state a preliminary theorem about the structure of $W=\mathbb{C}[Q]$  and the subalgebras $\mathcal{D}_{z,x}$ when this is the case.    This is the theorem that allows a local classification of a QCA as a QLGA. 
 Let $\mathcal{D}_{{x-y},x}$ be as defined in~\eqref{Projectiondef}, where  $x \in \mathbb{Z}^n$, and $y \in  \mathcal{N}$.    
\begin{thm}\label{thmS}
 Suppose that  $R$ is the  global evolution of a QCA with neighborhood $\mathcal{N}$. 
Then
$\mathcal{A}_x = \begin{rm}{span }\end{rm}(\prod_{y \in  \mathcal{N}} \mathcal{D}_{{x-y},x})$
if and only if there exists an isomorphism of vector spaces:
\begin{equation*}
S :  W \longrightarrow \bigotimes_{z \in \mathcal{N}} V_{z}
\end{equation*}
for some vector spaces $\{V_{z}\}_{z \in \mathcal{N}}$.  Under the isomorphism $S$, for each $y   \in \mathcal{N}$:
\begin{equation*}
\mathcal{D}_{{x-y},x} \cong  \begin{rm}{End}\end{rm}(V_y) \otimes \bigotimes_{z \in \mathcal{N}, z \neq y} \mathbb{I}_{V_z}
\end{equation*}
\end{thm}

Before proving the theorem, let us state  a corollary describing the structure of the algebras $\mathcal{R}_x$ under the conditions  of the theorem. These algebras $\mathcal{R}_x$ are the images of algebras $\mathcal{A}_x$ localized on a single cell, after one timestep of the global evolution $R$. 

\begin{cor} \label{corS} Suppose that $R$ is the  global evolution of a QCA  with neighborhood $\mathcal{N}$, and satisfies   $\mathcal{A}_x = \begin{rm}{span }\end{rm}(\prod_{y \in  \mathcal{N}} \mathcal{D}_{{x-y},x})$. Then:
 
\begin{enumerate}[label=(\roman{*})] 
\item \label{cor1}  $\mathcal{A}_x  = \begin{rm}{End}\end{rm}(W)  \cong \bigotimes_{z \in \mathcal{N}} \begin{rm}{End}\end{rm}(V_z)$, for all $x \in \mathbb{Z}^n$.
\item \label{cor2}  The dimension of  $W$, $d_W$, is a product of the dimensions of $V_z$, $d_{V_z}$, i.e.,  $d_W =  \prod_{z \in \mathcal{N}} d_{V_z}$.
\item \label{cor3}  $\mathcal{R}_x = \begin{rm}{span }\end{rm}(\prod_{k \in  \mathcal{N}} \mathcal{D}_{x,{x+k}}) \cong  \bigotimes_{k \in  \mathcal{N}}\begin{rm}{End}\end{rm}(V_{k})$.
\end{enumerate}
\end{cor}
\begin{proof} \ref{cor1} and~\ref{cor2} are obvious from the  isomorphism $S$ in   Theorem~\ref{thmS}. For the proof of~\ref{cor3}, we note that by definition of $\mathcal{D}_{x,{x+k}}$~\eqref{Projectiondef}, $\mathcal{R}_x \supset  \begin{rm}{span }\end{rm}(\prod_{k \in  \mathcal{N}} \mathcal{D}_{x,{x+k}})$. But from  Theorem~\ref{thmS}, $\begin{rm}{span }\end{rm}(\prod_{k \in  \mathcal{N}} \mathcal{D}_{x,{x+k}}) \cong \bigotimes_{k \in  \mathcal{N}}\begin{rm}{End}\end{rm}(V_{k}) \cong \begin{rm}{End}\end{rm}(W)$. We already know  that $\mathcal{R}_x  =  R^\dag \mathcal{A}_x R \cong  \begin{rm}{End}\end{rm}(W)$. Hence the conclusion.
\end{proof}

This corollary shows that if the conditions of Theorem~\ref{thmS} are met then the algebra $\mathcal{A}_x$ is isomorphic to an algebra generated by a set of commuting algebras each of which is the endomorphisms (full matrix algebra) of a single tensor factor of the space $W$. 

The following  lemma  shows that if there is a  set of pairwise commuting self-adjoint algebras that generate $\begin{rm}{End}\end{rm}(W)$, then there is a tensor product decomposition of the space $W$ such that each of the commuting self-adjoint  algebras is the set of endomorphisms (full matrix algebra) of one of the tensor factors. This  will be useful in the proof of Theorem~\ref{thmS}. We note that the case with a pair of such algebras has been treated in Theorem $3$ in~\cite{anw:odqca} and the proof given there has been derived based on the results in~\cite{g:masdptc}. We present a different proof of the more general case primarily based on the material  in~\cite{wall:sri},  included in Appendix~\ref{apprt} and~\ref{saalg}.
\begin{lemma} \label{lemmaS}
Let  $\{M_k : M_k \subset \begin{rm}{End}\end{rm}(W)\}^m_{k =1}$, be a finite set of distinct pairwise commuting self-adjoint algebras, each  containing the identity operator, such that  $\begin{rm} span  \end{rm}(\prod^m_{k =1} M_k) = \begin{rm}{End}\end{rm}(W)$.   Then there is a vector space isomorphism:
\begin{equation*}
S_W :    W  \longrightarrow  \bigotimes^m_{k =1} V_k
\end{equation*}
for some vector spaces $\{V_k\}^m_{k =1}$. Under the isomorphism $S_W$,  for each $1 \leq y \leq m$:
\begin{equation*}
M_y  \cong  \begin{rm}{End}\end{rm}(V_y) \otimes \bigotimes_{k  \in \{1,\ldots,m\}\setminus \{y\}} \mathbb{I}_{V_k}  
\end{equation*}
\end{lemma}
\begin{proof} If $m=1$, there is nothing to prove. So  assume $m > 1$. Consider $M_m$ from the set $\{M_k : M_k \subset \begin{rm}{End}\end{rm}(W)\}^m_{k=1}$. Let ${T_m} = \text{span}(\prod^{m-1}_{t=1} M_t)$. Then we know that  ${T_m}$ and $M_m$ are mutually commuting self-adjoint subalgebras,  each  containing the identity operator. By Proposition~\ref{propSA},  $W$ is then a completely reducible ${T_m}$-module.  Corollary~\ref{corSS} then implies that there exists a finite set  of irreducible ${T_m}$-modules (here the index set $\{j : 1\leq j \leq r\}$ enumerates  the finite (by Corollary~\ref{corSS})  set  of equivalence classes of irreducible representations of ${T_m}$, denoted   $\widehat{T}_m$.  Each value of $j$ corresponds to a distinct class $\lambda \in \widehat{T}_m$) that we denote $\{V_j \}^r_{j=1}$, such that  if we denote the set of their multiplicity spaces: $\{U_j =  \begin{rm}{Hom}\end{rm}_{{T_m}}(V_j,W)\}^r_{j=1}$, then we have    a   ${T_m}$-module isomorphism $S_m$: 
\begin{align*} 
S_m: \bigoplus^r_{j=1}  U_j \otimes V_j &\longrightarrow     W \\
\sum^r_{j=1}  u_j \otimes  v_j &\mapsto  \sum^r_{j=1}  u_j( v_j) \nonumber
\end{align*}
Under this isomorphism:
\begin{equation*}
T_m = S_m\big(\bigoplus^r_{j=1}\mathbb{I}_{U_j} \otimes \begin{rm}{End}\end{rm}(V_j)   \big) {S}^{-1}_{m} 
\end{equation*}
\begin{equation*}
\text{Comm}({T_m}) = S_m \big( \bigoplus^r_{j=1} \begin{rm}{End}\end{rm}(U_j)   \otimes  \mathbb{I}_{V_j} \big) S_m^{-1} 
\end{equation*}
where $\text{Comm}({T_m})$ is the commutant of $T_m$, as defined  by~\eqref{commdef}, and  each $U_j$ is   a  $\text{Comm}({T_m})$-module under the action given by~\eqref{bmodact}. Since, by definition of $T_m$, $M_m \subset \text{Comm}({T_m})$, this implies:
\begin{equation*} 
\text{span}(M_m {T_m}) \subset \text{span}(\text{Comm}({T_m}) {T_m})  = S_m \big(\bigoplus^r_{j=1}\begin{rm}{End}\end{rm}(V_j)  \otimes \begin{rm}{End}\end{rm}(U_j)\big) S_m^{-1} \subset \text{End}(W) 
\end{equation*}
But by assumption, $\text{span}(M_m {T_m}) = \text{End}(W)$. Then the inclusions above are equalities, and there must  only be one summand in the sum above. That is, there exists some   irreducible $T_m$-module  $Y_m$, such that  if we denote  its multiplicity space: $
 U_m =  \begin{rm}{Hom}\end{rm}_{{T_m}}(Y_m,W)$,
then we have   a $T_m$-module   isomorphism $S_m$:
\begin{align} \label{Smisom}
S_m:   U_m \otimes Y_m &\longrightarrow    W \\ 
u_m\otimes  y_m &\mapsto  u_m(y_m) \nonumber
\end{align}
Furthermore, under $S_m$:
\begin{equation} \label{Tmmin1act}
{T_m}  \cong  \mathbb{I}_{U_m}  \otimes   \begin{rm}{End}\end{rm}(Y_m) 
\end{equation}
and:
\begin{equation*}
 M_m \cong   \begin{rm}{End}\end{rm}(U_m) \otimes  \mathbb{I}_{Y_m}
\end{equation*}
Here the subscript $m$  signifies that $Y_m$ is the only class of irreducible $T_m$-module in $W$, and does not correspond to the indexing $\{j : 1\leq j \leq r\}$ used above. This subscripting convention will be used for general  $k$ in the set:  $\{k: m-1 \geq k \geq 2\}$, in the arguments to follow.

In particular this proves the lemma for $m=2$: if $m=2$, observing that  $T_2 = M_1$, we can define the vector spaces, $V_2 := U_2$, and $V_1 := Y_2$, and  the isomorphism $S_W :=  {S}^{-1}_2$. Then the vector spaces $V_1$ and $V_2$, and the isomorphism $S_W$  satisfy the statement of the lemma. 

Now assume $m > 2$. Since $T_m$ is a subalgebra of $\begin{rm}{End}\end{rm}(W)$, we choose the irreducible $T_m$-module $Y_m$ in~\eqref{Smisom} to be a subspace of $W$:  $Y_m \subset W$.

Clearly the  subalgebras in the set $\{M_k\}^{m-1}_{k=1}$ are subalgebras of $T_m = \text{span}(\prod^{m-1}_{t=1} M_t)$. Also, we have that $T_m = \begin{rm}{End}\end{rm}(Y_m)$, by the definition of the  $T_m$ action on $Y_m$, and~\eqref{Tmmin1act}.
Let us consider the restricted actions of the  subalgebras in the set $\{M_k\}^{m-1}_{k=1}$ on $Y_m$. These restricted actions    comprise a set of   mutually commuting, self-adjoint  (for the inner product on $W$ restricted to $Y_m$) subalgebras of $T_m=\begin{rm}{End}\end{rm}(Y_m)$, each of which contains the identity operator on $Y_m$.        

We proceed by induction for  this part of the argument.   Let $k$ be in the set: $m-1 \geq k \geq 3$. Let $T_{k} = \text{span}(\prod^{k-1}_{t=1} M_t)$.      By way of induction, assume that there exists a descending sequence of subspaces of $Y_m$, denoted $\{Y_{k}\}^{m-1}_{k=3}$ satisfying: $Y_m\supset Y_{m-1} \supset \ldots \supset Y_3$, such that each $Y_k$ is an irreducible  $T_k$-module under the restricted $T_k$  action.   Since $ {T_m} \supset \ldots \supset T_3$ is a descending sequence of subalgebras, therefore,    ${T_{k}}$  can be restricted to act  on $Y_{k+1}$.  Let $U_{k}  =  \begin{rm}{Hom}\end{rm}_{T_{k}}(Y_{k},Y_{k+1})$.  Also assume  that    there is a  $T_{k}$-module isomorphism:
\begin{align} \label{Stisom}
{S}_k :   U_{k} \otimes Y_{k} &\longrightarrow    Y_{k+1} \\
u_{k} \otimes  y_{k} &\mapsto  u_{k}( y_{k}) \nonumber
\end{align}
where $u_{k} \in U_{k}$, and $y_{k} \in Y_{k}$. 
Under this isomorphism, assume that: 
\begin{equation} \label{Tkact}
{T_{k}}  \cong  \mathbb{I}_{U_{k}}  \otimes   \begin{rm}{End}\end{rm}(Y_{k}) 
\end{equation}
and:
\begin{equation*}
 M_{k} \cong   \begin{rm}{End}\end{rm}(U_{k}) \otimes  \mathbb{I}_{Y_{k}}
\end{equation*}

Now let us consider the terminating case when $k=2$.     The  subalgebras  $M_1$, $M_2$  are subalgebras of $T_3 = \text{span}(M_1 M_2)$. By induction assumption~\eqref{Tkact}, under restricted $T_3$ action, $T_{3} = \begin{rm}{End}\end{rm}(Y_{3})$. Restricted actions of the subalgebras $M_1$, $M_2$   comprise a set of   mutually commuting, self-adjoint  (for the inner product on $W$ restricted to $Y_{3}$) subalgebras of $T_{3} = \begin{rm}{End}\end{rm}(Y_{3})$, and each of them contains the identity operator on $Y_3$.   Also, $T_2 = M_1$. This case then reduces to the case $m=2$, that we have already proved before stating the  induction assumption.  The $m=2$ result then implies that   for the restricted $M_1$ action   on $Y_3$,  there is some irreducible $M_1$-module  $Y_{2}$, such that  if we denote  its multiplicity space: $U_{2} =  \begin{rm}{Hom}\end{rm}_{M_1}(Y_{2},Y_{3})$, we have  an $M_1$-module isomorphism ${S}_2$:
\begin{align} \label{S1isom}
{S}_2 :   U_{2} \otimes Y_{2} &\longrightarrow    Y_3 \\
u_{2} \otimes  y_{2} &\mapsto  u_{2}( y_{2}) \nonumber
\end{align}
where $u_{2} \in U_{2}$, and $y_{2} \in Y_{2}$.  Under the isomorphism ${S}_2$, we have  that: 
\begin{equation*}
M_1 \cong  \mathbb{I}_{U_{2}}  \otimes   \begin{rm}{End}\end{rm}(Y_{2}) 
\end{equation*}
and:
\begin{equation*}
 M_2 \cong   \begin{rm}{End}\end{rm}(U_{2}) \otimes  \mathbb{I}_{Y_{2}}
\end{equation*} 

Using the isomorphisms~\eqref{Stisom} of the induction step, with~\eqref{Smisom}, and~\eqref{S1isom},  implies that $S_m$ is given as:
\begin{align*} 
S_m:   U_m\otimes U_{m-1}  \otimes \cdots \otimes U_{2} \otimes Y_2 &\longrightarrow    W \\
u_{m}\otimes u_{m - 1} \otimes \cdots \otimes u_{2} \otimes y_2 &\mapsto  u_m(u_{m-1}(\ldots (u_2(y_2)))) \nonumber
\end{align*}
where $u_{k} \in U_{k}$,  for $2 \leq k \leq m$, and $y_{2} \in Y_{2}$. 

Define the vector spaces $V_k := U_{k}$, for $2 \leq k \leq m$, and $V_1 := Y_2$, and the isomorphism $S_W := {S}^{-1}_m$.  Then the vector spaces $\{V_k\}^m_{k=1}$, and the isomorphism $S_W$  satisfy the statement of the lemma.  The proof shows that the order in which  the subalgebras $\{M_k\}^{m}_{k=1}$ are enumerated does not affect the conclusion. One can also prove this lemma via a slightly more general route, noting that  the subalgebras involved are semisimple,  which follows by Corollary~\ref{sass},  and by using the semisimple version of the Double Commutant Theorem, Theorem~\ref{dblctthm}. The approach  in the given proof is more specific.

\end{proof}
\begin{remarknonum} We observe that if some of the subalgebras $M_k$ in  Lemma~\ref{lemmaS} are multiples of the identity operator, then the corresponding vector spaces $V_k$ will be (trivial) one-dimensional spaces.
\end{remarknonum}

\begin{proof}[Proof of Theorem \ref{thmS}]
 
Let us assume that   $\mathcal{A}_x = \begin{rm}{span }\end{rm}(\prod_{y \in  \mathcal{N}} \mathcal{D}_{{x-y},x})$.     Fix $x \in \mathbb{Z}^n$, let $y \in \mathcal{N}$.  The subalgebras $\mathcal{D}_{x-y,x}$ are   self-adjoint (as $\mathcal{R}_{x-y}$ are self-adjoint) subalgebras of $\mathcal{A}_x$,  and they pairwise commute (as $\mathcal{R}_{x-y}$  pairwise commute), and each contains the identity of $\mathcal{A}_x$. Hence, by Lemma~\ref{lemmaS},  there exists a  isomorphism of vector spaces:
\begin{equation*}
S_x :   W \longrightarrow \bigotimes_{k \in \mathcal{N}} V_{x-k,x}
\end{equation*}
for some vector spaces $\{V_{x-k,x}\}_{k \in \mathcal{N}}$. Under the isomorphism $S_x$, for all  $k \in \mathcal{N}$:
\begin{equation*}
\mathcal{D}_{x-k,x} \cong  \begin{rm}{End}\end{rm}(V_{x-k,x}) \otimes \bigotimes_{z \in \mathcal{N}, z \neq k} \mathbb{I}_{V_{x-z,x}} 
\end{equation*}

By translation invariance of $R$,   the isomorphisms $S_x$ can be taken to be the same and  the sets of vector spaces $\{V_{x-k,x}\}_{k \in \mathcal{N}}$,  for all $x \in \mathbb{Z}^n$,  can be taken to be the same. Therefore, there exists a   isomorphism of vector spaces, $S$:
\begin{equation*}
S :   W \longrightarrow \bigotimes_{z \in \mathcal{N}} V_{z}
\end{equation*}
for some vector spaces $\{V_{z}\}_{z \in \mathcal{N}}$. Under the isomorphism $S$, for all  $y \in \mathcal{N}$:
\begin{equation*}
\mathcal{D}_{x-y,x} \cong  \begin{rm}{End}\end{rm}(V_y) \otimes \bigotimes_{z \in \mathcal{N}, z \neq y} \mathbb{I}_{V_{z}} 
\end{equation*}

For the converse, let us assume that  there exists, for some vector spaces $\{V_{z}\}_{z \in \mathcal{N}}$,  an isomorphism $S$: 
\begin{equation*}
S :   W \longrightarrow \bigotimes_{z \in \mathcal{N}} V_{z}
\end{equation*}
such that
\begin{equation*}
\mathcal{D}_{x-y,x} \cong  \begin{rm}{End}\end{rm}(V_y) \otimes \bigotimes_{z \in \mathcal{N}, z \neq y} \mathbb{I}_{V_{z}} 
\end{equation*}
then:
\begin{align*}
{\rm span}\big(\prod_{y\in\mathcal{N}}\mathcal{D}_{x-y,x}\big) &=  S^{-1}{\rm span}\bigg(\prod_{y\in\mathcal{N}}\big( \begin{rm}{End}\end{rm}(V_y)  \otimes \bigotimes_{z \in \mathcal{N}, z \neq y} \mathbb{I}_{V_{z}} \big)\bigg)S\\
&=S^{-1}\big(\bigotimes_{z \in \mathcal{N}} \begin{rm}{End}\end{rm}(V_z)\big)S = \mathcal{A}_x
\end{align*}
\end{proof}

Theorem~\ref{thmS} establishes an equivalence between the fact that the algebra of operators on a single cell is generated by those pieces of evolved subalgebras that are localized on the cell, and a tensor product structure of the single cell Hilbert space. In the following we show that Theorem~\ref{thmS} enables the classification of a subset of QCAs as QLGAs. 

%%%%%%%%%%%%%%%%%%%%%%%%%%%%%%%%%%%%%%%%%%%%%%%%%%%%%%%%%%%%%%%%%%%%%%
\subsection{Locally equivalent QCA} \label{loceqv}
We note that, just as changing the specific choice of symbols in the alphabet  of a classical CA does not change the dynamics, any local isomorphism of the cell Hilbert space and the global time evolution operator of a QCA results in equivalent dynamics. The local isomorphisms of the cell Hilbert space are simply the unitary group of the appropriate dimension. It is therefore only dimension $d$ of the cell Hilbert space that defines the local Hilbert space $W$ because all Hilbert spaces of the same finite dimension $d$ are isomorphic to $\mathbb{C}^d$. For any QCA we may therefore define a local equivalence class of QCA defined by a local Hilbert space $W$ of finite dimension, all local unitary operators on $W$ and a global evolution $R$ on $\mathcal{H}_{\mathcal{C}}$. 

Let us define a local transformation of the Hilbert space of finite configurations. Given a local unitary operator $U$  on $W$,   the image of the symbol basis $B_Q$~\eqref{BQ}, denoted $B_{Q,U}$,  is:
\begin{equation} \label{BQU}
B_{Q,U} = U(B_Q) = \{U\ket{q} : q\in Q\}
\end{equation}
One observes that the  inner product on $B_Q$    is preserved on the image $B_{Q,U}$ by the local unitary transformation $U$: $\langle U q \vert U q' \rangle = \langle q \vert q' \rangle$,  for all $\ket{q}, \ket{q'} \in B_Q$.

A {\em set of locally transformed finite configurations} $\mathcal{C}_U$ is as defined in Definition~\ref{sofc} with $B_Q$ replaced by $B_{Q,U}$. Then the inner product on $W$~\eqref{sofcip} restricted to $B_{Q,U}$, naturally induces an inner product on $\mathcal{C}_U$, hence on $\text{span}(\mathcal{C}_U)$. A {\em locally transformed Hilbert space of finite configurations}, denoted by $\mathcal{H}_{\mathcal{C}_U}$, is the ${\ell}^2$ completion of $\text{span}(\mathcal{C}_U)$, as in Definition~\ref{hofc},  under the  norm induced by the inner product on $\text{span}(\mathcal{C}_U)$. 

Let $\tilde U$ be the local transformation that maps $\mathcal{H}_{\mathcal{C}}$ to $\mathcal{H}_{\mathcal{C}_U}$:
\begin{equation*} 
\tilde U:  {\mathcal{H}}_{\mathcal{C}} \longrightarrow \mathcal{H}_{\mathcal{C}_U}
\end{equation*}
This map is defined on $\mathcal{C}$ as:
\begin{align*} %\label{Udefn}
\tilde U: \mathcal{C} &\longrightarrow \mathcal{C}_U\\
\bigotimes_{x \in \mathbb{Z}^n} \vert {c}_{x} \rangle  &\mapsto \bigotimes_{x \in \mathbb{Z}^n} U(\vert {c}_{x} \rangle) \nonumber
\end{align*}
and  extended to  ${\mathcal{H}}_{\mathcal{C}}$. 

Given a  global evolution operator $R$ on $\mathcal{H}_{\mathcal{C}}$ (we assume  the neighborhood $\mathcal{N}$  in the definition of $R$), a locally transformed global evolution operator $R_U$ on $\mathcal{H}_{\mathcal{C}_U}$ is defined as:
\begin{equation} \label{RUdef}
R_U := {\tilde U}  R {\tilde U}^\dag
\end{equation}
The set of locally transformed Hilbert spaces of finite configurations is: 
\begin{equation*}
L_W := \{\mathcal{H}_{\mathcal{C}_U} :  U \text{ is a unitary operator on } W\}
\end{equation*}
Now consider the set of pairs
\begin{equation*}
E_W := \{(R,\mathcal{H}_{\mathcal{C}}) :  \mathcal{H}_{\mathcal{C}} \in L_W,  R \text{ is a global evolution operator on } \mathcal{H}_{\mathcal{C}}\}
 \end{equation*}
We say that the pair $(R,\mathcal{H}_{\mathcal{C}}) \in E_W$ is \textit{equivalent} to $(R',\mathcal{H}_{\mathcal{C}'}) \in E_W$, if $\mathcal{H}_{\mathcal{C}'} =  \mathcal{H}_{\mathcal{C}_U}$, and $R' = R_U$, for some unitary operator $U$ on $W$. This defines an equivalence relation on the set of pairs  $E_{W}$. Therefore,  the pair   $(R,\mathcal{H}_{\mathcal{C}})$ is a representative of the  equivalence class $[(R,\mathcal{H}_{\mathcal{C}})] =\{(R_U,\mathcal{H}_{\mathcal{C}_U}) :  U \text{ is a unitary operator on } W\}$. 

\subsection{Separability of the quiescent symbol and the propagation operator} \label{sqsbl}

Assume there is an isomorphism $S$ as in Theorem~\ref{thmS}.  Denote $\hat W = S(W) = \bigotimes_{z \in \mathcal{N}} V_z$.   
$\hat W$ can   be identified with $\mathbb{C}^{d_{W}}$ and imbued with an inner product. We choose this inner product such that the image of the symbol basis $B_Q$~\eqref{BQ}  under $S$, $S(B_Q) = \{S(\vert q \rangle) : q \in Q\}$ is an orthonormal set. With this inner product, $\hat W$ becomes a Hilbert space. Note that this basis of $\hat W$, which is the image of $B_Q$, will not in general consist of simple tensors.  In addition, imbuing each $V_k$ with an inner product and then constructing an inner product on $\hat W$ may result in an inner product on $\hat W$ which is inconsistent with the requirement that the image of the symbol basis be orthonormal.

For each $z\in\mathcal{N}$, let: 
\begin{equation} \label{Vnu0basisdef}
\{\vert k_z \rangle\}_{1\leq k_z \leq d_{V_z}}
\end{equation}
be a basis of $V_{z}$.  Use  $\textbf{k}$ to label   elements of  the Cartesian product $\prod_{z \in \mathcal{N}} \{1,..,d_{V_{z}}\}$.  A natural basis for the Hilbert  space $\hat W = \bigotimes_{z \in \mathcal{N}} V_z$ can be written:
\begin{equation} \label{Bnu0def}
B_{\mathcal{N}} = \{ \bigotimes_{z \in \mathcal{N}} \vert \textbf{k}(z) \rangle : \textbf{k} \in \prod_{z \in \mathcal{N}} \{1,..,d_{V_{z}}\} \}
\end{equation}
The notation $\textbf{k}(z)$ stands for the $z$-th coordinate element of  $\textbf{k} \in \prod_{z \in \mathcal{N}} \{1,..,d_{V_{z}}\}$. And where:
\begin{equation} \label{kzdef}
\vert \textbf{k}(z) \rangle   \in  V_z
\end{equation}
We shall refer to an individual tensor factor within a cell, i.e., an element of $V_z$ for some $z \in \mathcal{N}$,  as a \textit{component}.    The notation here is just as for the case of a classical LGA in equation~(\ref{classk}). In the classical case different lattice vectors $z$ can have different alphabets. Here, the dimensions of the $V_z$ may differ, and hence they have distinct bases. In the case where $V_z=\mathbb{C}^2~\forall~z$ the basis of $\hat W$ that we obtain corresponds to the occupation number representation for systems of particles. This representation is used widely, particularly in quantum simulation of fermionic systems~\cite{Zanardi, qs:science, JSPJL}.

We add $\vert \hat{q}_0 \rangle= S(\vert q_0 \rangle)$ to the basis $B_{\mathcal{N}}$ and define:
\begin{equation} \label{hatQdef}
\hat Q = B_{\mathcal{N}} \cup \{\vert \hat{q}_0 \rangle \}.
\end{equation}
This is necessary because $\vert \hat{q}_0 \rangle$ may not be an element of the basis $B_{\mathcal{N}}$, in which case we can express it as:
\begin{equation*}
\vert {\hat{q}_0} \rangle  = \sum_{\substack{\textbf{k} \in \prod_{j \in \mathcal{N}} \{1,..,d_{V_j}\} }} \omega({\textbf{k}}) \bigotimes_{z \in \mathcal{N}} \vert \textbf{k}(z) \rangle  
\end{equation*}
where  $\textbf{k} \in \prod_{z \in \mathcal{N}} \{1,..,d_{V_{z}} \}$ and $\omega$ is some complex valued function:
\begin{equation*}
\omega:  \prod_{z \in \mathcal{N}} \{1,..,d_{V_{z}}\} \longrightarrow \mathbb{C}
\end{equation*}

We  then  consider the set of finite configurations, taken over every site, and define the counterparts of the set of finite configurations and the Hilbert space of finite configurations:
\begin{defn}\label{Chatsofc}
The \textit{set of finite configurations in component form}, denoted by $\hat{\mathcal{C}}$, is  the set of  simple tensor products   with only  finitely many active elements, 
\begin{equation*} \hat{\mathcal{C}} := \{ \bigotimes_{x \in \mathbb{Z}^n} \vert {\hat c}_{x} \rangle :  \vert {\hat c}_{x} \rangle  \in \hat Q,  \text{ all but finite } \vert {\hat c}_{x} \rangle  = \vert \hat{q}_0 \rangle \}
\end{equation*}
 \end{defn}

 The inner product on $\hat{\mathcal{C}}$ is induced by the inner product on $\hat W$. Let $\vert \hat c \rangle = \bigotimes_{x \in \mathbb{Z}^n} \vert {\hat c}_x \rangle, \vert {\hat c}' \rangle = \bigotimes_{x \in \mathbb{Z}^n} \vert {{\hat c}_x}' \rangle \in \hat{\mathcal{C}}$. Define the inner product of the elements $\vert {\hat c} \rangle$, $\vert {\hat c}' \rangle$:
\begin{equation*}
\langle {\hat c} \vert {\hat c}' \rangle = \prod_{x \in \mathbb{Z}^n} \langle {\hat c}_x \vert {{\hat c}_x}' \rangle
\end{equation*}
and extend it by linearity to get an inner product on  span($\hat{\mathcal{C}}$).

\begin{defn}\label{Chathofc}
The \textit{Hilbert space of finite configurations in component form},   denoted by  ${\mathcal{H}}_{\hat{\mathcal{C}}}$, is  the ${\ell}^2$ completion of $\text{span}(\hat{\mathcal{C}})$, under the norm  induced by the above inner product. 
\end{defn}

We will use the term \textit{Hilbert space  of an  individual cell} $x$ to  simply mean the set in which $\vert {\hat{c}}_{x} \rangle$ can  take any value, without violating the definition of $\hat{\mathcal{C}}$. The tensor factors corresponding to the Hilbert spaces  of individual cells  are indexed by the  \textit{cell index} $x \in \mathbb{Z}^n$. Thus we can index the Hilbert space of a cell $x$ by ${\hat W}_x := \hat W$. These spaces themselves are composed of tensor factors $V_z$ indexed by the \textit{component index} $z$, which runs over the neighborhood $\mathcal{N}$, i.e., $z \in \mathcal{N}$. So we index the sub-factors of ${\hat W}_x$ by $V_{z,x} := V_z$. We also need for every cell $x$, the  component coordinates  given by $\textbf{k}_x \in \prod_{z \in \mathcal{N}} \{1,..,d_{V_{z}}\}$, as in~\eqref{Bnu0def}. Then, as in~\eqref{kzdef}, $\ket{{\bf k}_x(z)}$ is a basis element of $V_{z,x}$. Also, we have that  ${\hat W}_x = \bigotimes_{z \in \mathcal{N}} V_{z,x}$. 

The picture of time evolution that we shall establish is that of a permutation on tensor factors that rearranges the components $\ket{{\bf k}_x(z)}$ among the cells $x \in \mathbb{Z}^n$,  followed by transformations acting only on each cell Hilbert space ${\hat W}_x = \bigotimes_{z \in \mathcal{N}} V_{z,x}$. We first take care to define the map that permutes the tensor factors and understand what its existence implies  about the structure of the space ${\mathcal{H}}_{\hat{\mathcal{C}}}$. These properties will be useful in proving the main theorem.

We will formally specify  the  \textit{propagation operator}, relative to the neighborhood $\mathcal{N}$, by specifying it on elements of $\hat{\mathcal{C}}$. This map acts by shifting one  component from each of the  the basis elements $\bigotimes_{z \in \mathcal{N}} \vert \textbf{k}_y(z) \rangle$ of the neighborhood cells $y \in \mathcal{N}_x$ into the corresponding component of the cell $x$. 

We see at once that to be able to  define the propagation operator on $\hat{\mathcal{C}}$ will require that the new quiescent symbol $\vert \hat{q}_0 \rangle = S(\vert q_0 \rangle) \in \bigotimes_{z \in \mathcal{N}} V_z$, is a simple tensor, i.e.,  of the form:
\begin{equation*} 
\vert \hat{q}_0 \rangle = \bigotimes_{z \in \mathcal{N}} \vert \hat{q}_{0,z} \rangle
\end{equation*}
for some  $\vert \hat{q}_{0,z} \rangle  \in V_z$. If this is not the case then we may always find a unitary operator $U$  acting on $W$, such that $S(U \vert q_0 \rangle)$ is a simple tensor. In this case we work with the locally equivalent QCA defined by the pair $(R_U, \mathcal{H}_{\mathcal{C}_U})$ in~\eqref{RUdef} of Section~\ref{loceqv}. Let us therefore assume that $\vert \hat{q}_0 \rangle$ has this property.

Under this assumption, the basis for $V_z$ in~\eqref{Vnu0basisdef} can be chosen such that $\vert \hat{q}_0 \rangle \in B_{\mathcal{N}}$. With the inclusion of $\vert \hat{q}_0 \rangle$ in  $B_{\mathcal{N}}$, the definition of $\hat Q$~\eqref{hatQdef} becomes $\hat Q = B_{\mathcal{N}}$~\eqref{Bnu0def}.  This allows us to  write an element $\vert {\hat{c}}\rangle \in \hat{\mathcal{C}}$ in the form:
\begin{equation} \label{chattrans}
\vert {\hat{c}}\rangle = \bigotimes_{x \in \mathbb{Z}^n}  \vert {\hat{c}}_x \rangle = \bigotimes_{x \in \mathbb{Z}^n} \bigotimes_{z \in \mathcal{N}} \vert \textbf{k}_x(z) \rangle
\end{equation}
where the $x$th tensor factor of an element $\bigotimes_{x \in \mathbb{Z}^n}  \vert {\hat{c}}_x \rangle \in \hat{\mathcal{C}}$ is given by:
\begin{equation} \label{tensorindex}
\vert {\hat{c}}_x \rangle  = \bigotimes_{z \in \mathcal{N}} \vert \textbf{k}_x(z) \rangle \in {\hat W}_x
\end{equation}
and where:
\begin{equation}  \label{componentindex}
\vert \textbf{k}_x(z) \rangle   \in V_{z,x} 
\end{equation}

Now we may define the  \textit{propagation operator relative to the neighborhood $\mathcal{N}$}, that we denote by $\sigma$. We combine the  cell and  component indices together in a  pair $(x,z) \in \mathbb{Z}^n \times \mathcal{N}$. As $\sigma$ permutes the cell indices $x \in \mathbb{Z}^n$, leaving the component indices $z \in \mathcal{N}$ intact, we can define $\sigma$,  on $\hat{\mathcal{C}}$, in terms of a cell-indexing function $s_\sigma: \mathbb{Z}^n \times \mathcal{N} \longrightarrow \mathbb{Z}^n$,  as:
  \begin{equation}  \label{tfdefin}
 \left . \begin{array} {cccc}
 \sigma : &\hat{\mathcal{C}}&\longrightarrow  & \hat{\mathcal{C}} \\
  &\bigotimes_{x \in \mathbb{Z}^n} \bigotimes_{z \in \mathcal{N}} \vert \textbf{k}_x(z) \rangle  &\mapsto  &\bigotimes_{x \in \mathbb{Z}^n} \bigotimes_{z \in \mathcal{N}} \vert \textbf{k}_{s_\sigma(x,z)}(z) \rangle
    \end{array}
  \right .
\end{equation} 
where $s_\sigma$ is defined as follows. 
\begin{align*}
s_\sigma: \mathbb{Z}^n \times \mathcal{N} &\longrightarrow \mathbb{Z}^n \\
(x,z) &\mapsto x+z
\end{align*}
In other words $\sigma$  moves the $z$-th component of  cell $x+z$  to the $z$-th component  of cell $x$.
 $\sigma$ is extended to a unitary, translation-invariant operator on ${\mathcal{H}}_{\hat{\mathcal{C}}}$. We can omit $s_\sigma$ and write the action of $\sigma$ as:
 \begin{equation*}  
 \left . \begin{array} {cccc}
 \sigma : 
  \bigotimes_{x \in \mathbb{Z}^n} \bigotimes_{z \in \mathcal{N}} \vert \textbf{k}_x(z) \rangle  \mapsto  \bigotimes_{x \in \mathbb{Z}^n} \bigotimes_{z \in \mathcal{N}} \vert \textbf{k}_{x+z}(z) \rangle
    \end{array}
  \right .
\end{equation*}

We note that the structure described above corresponds to the classical case in the following way. We identify the states of a classical lattice-gas with the elements of $\hat{\mathcal{C}}$. The state of a lattice-gas site is simply $\vert \hat c_x\rangle$, and the state of lattice vector $z$ is an element of the basis of $V_z$. We may think, therefore, of the tensor factors $V_z$, as Hilbert spaces of lattice vectors pointing to neighborhood site $z$. Note however that because these are quantum models arbitrary superposition states over the elements of this basis are allowed.

\subsection{The collision operator}

As we shall prove below, an evolution composed of the cell-wise isomorphism $S$ and the propagation operator $\sigma$ obeys the conditions to be a QCA. However, this is not the only possible evolution. Performing any cell-wise unitary after the isomorphism $S$ will also give us a QCA. 

We therefore formally construct  an operator $\hat F$, by  specifying it  on the basis $\hat{\mathcal{C}}$ of ${\mathcal{H}}_{\hat{\mathcal{C}}}$ in terms of some unitary operator $F$  on   $\hat W = \bigotimes_{z \in \mathcal{N}} V_z$, as follows:
\begin{align*}
\hat F: \bigotimes_{x \in \mathbb{Z}^n} \vert {\hat c}_{x} \rangle  &\mapsto \bigotimes_{x \in \mathbb{Z}^n} F(\vert {\hat c}_{x} \rangle)
\end{align*}
Then we have the following lemma that tells us what conditions ensure that $\hat F$ is defined on ${\mathcal{H}}_{\hat{\mathcal{C}}}$:
\begin{lemma} \label{Fdefm} 
$\hat F$  is defined on  $\hat{\mathcal{C}}$ as a map $\hat F: \hat{\mathcal{C}}  \longrightarrow {\mathcal{H}}_{\hat{\mathcal{C}}}$, and   can be extended to  ${\mathcal{H}}_{\hat{\mathcal{C}}}$ as unitary  and  translation-invariant operator, if and only if   $F$ has  $\vert \hat{q}_0 \rangle$ as an invariant (an eigenvector with eigenvalue one): $F \vert \hat{q}_0 \rangle =  \vert \hat{q}_0 \rangle$.
 \end{lemma}
The proof of this lemma is  included in  Appendix~\ref{ApplF}. Suppose $F \vert \hat{q}_0 \rangle =  \vert \hat{q}_0 \rangle$, then we can define:
\begin{equation*} 
\hat F:  {\mathcal{H}}_{\hat{\mathcal{C}}} \longrightarrow {\mathcal{H}}_{\hat{\mathcal{C}}}
\end{equation*}
which is the extension of the following definition on $\hat{\mathcal{C}}$ to ${\mathcal{H}}_{\hat{\mathcal{C}}}$.
\begin{align} \label{Fdefn}
\hat F:  \hat{\mathcal{C}} &\longrightarrow {\mathcal{H}}_{\hat{\mathcal{C}}}\\
\bigotimes_{x \in \mathbb{Z}^n} \vert {\hat c}_{x} \rangle  &\mapsto \bigotimes_{x \in \mathbb{Z}^n} F(\vert {\hat c}_{x} \rangle) \nonumber
\end{align}

Whereas the operator $\sigma$ is a map that moves the tensor factors $V_z$ between cells of the automata, the operator $F$ acts locally on states in $\bigotimes_{z\in\mathcal{N}} V_z$. In the previous subsection we identified the tensor factors $V_z$ with lattice-gas vectors. The operator $F$ acts separately on the states of each lattice-gas site. This operator corresponds to the collision operator of a quantum lattice-gas, as we shall show in detail below.

\subsection{The definition of a Quantum Lattice-Gas Automaton (QLGA)}~\label{qlgadef}

At this point we have developed tools that are sufficient to give a formal definition of a Quantum Lattice-Gas (QLGA). The dynamics proceed by collision and propagation, as in the classical case and as in the models defined in~\cite{bib:meyer2,bib:meyer3,bib:meyer4,bib:meyer5} and~\cite{Bog1,Bog2,Bog3}. A QLGA   is defined as follows:
\begin{enumerate}[label=(\roman{*})]
\item{A lattice $\mathcal{L} = \mathbb{Z}^n$, whose sites we shall label by $x$.}
\item{A neighborhood $\mathcal{N}$. This translates to a neighborhood $\mathcal{N}_x$ for each cell $x$, consisting of a set of cells as follows:  $\mathcal{N}_x = \{x+z | z\in\mathcal{N}\}$.}
\item{For each element of the lattice  $\mathcal{L}$ there is a site Hilbert space $\hat W = \bigotimes_{z\in\mathcal{N}} V_z$, for some finite-dimensional vector spaces $\{V_{z}\}_{z \in \mathcal{N}}$.  $\hat W$ contains a distinguished unit vector, the \textit{quiescent state} $\vert \hat{q}_0 \rangle$, which is a simple tensor:
\begin{equation*} 
\vert \hat{q}_0 \rangle = \bigotimes_{z \in \mathcal{N}} \vert \hat{q}_{0,z} \rangle{\rm,~where}~\vert \hat{q}_{0,z} \rangle  \in V_z. 
\end{equation*}
}
\item A Hilbert space of finite configurations in component form ${\mathcal{H}}_{\hat{\mathcal{C}}}$ as defined in Definition~\ref{Chathofc}, in terms of $\hat W = \bigotimes_{z\in\mathcal{N}} V_z$ and  $\vert \hat{q}_0 \rangle$.
\item A state of the automaton, $\Psi$, is an element of the Hilbert space of finite configurations in component form.
\item A propagation operator relative to the neighborhood $\mathcal{N}$,    $\sigma: {\mathcal{H}}_{\hat{\mathcal{C}}} \mapsto {\mathcal{H}}_{\hat{\mathcal{C}}}$, as defined in~\eqref{tfdefin}. 
\item A local collision operator $F$, which is a unitary operator  on the site Hilbert space $\hat W = \bigotimes_{z\in\mathcal{N}} V_z$, such that $F$ has  $\vert \hat{q}_0 \rangle$ as an invariant (an eigenvector with eigenvalue one): $F \vert \hat{q}_0 \rangle =  \vert \hat{q}_0 \rangle$.  The associated  \textit{collision operator} $\hat F: {\mathcal{H}}_{\hat{\mathcal{C}}} \mapsto {\mathcal{H}}_{\hat{\mathcal{C}}}$, in terms of $F$,  as defined in~\eqref{Fdefn}.
\item A global evolution operator $\hat R$ which consists of applying the propagation operator $\sigma$ followed by the local collision operator  at every  site, i.e., is given by:
\begin{equation*}
\hat R   =    {\hat F} \sigma 
\end{equation*}
\end{enumerate}

The tensor factors $V_z$ may be identified with the lattice vectors of a classical lattice-gas, which point to the elements of the neighborhood of a site.

\subsection{Local inner product preserving transformation to component form} Let  $\tilde S$ be the  map that takes  ${\mathcal{H}}_{\mathcal{C}}$  to ${\mathcal{H}}_{\hat{\mathcal{C}}}$, i.e. the global map that is the result of applying the local isomorphism $S$ to every cell.
\begin{equation*} 
\tilde S:  {\mathcal{H}}_{\mathcal{C}} \longrightarrow {\mathcal{H}}_{\hat{\mathcal{C}}}
\end{equation*}
this map is defined on $\mathcal{C}$ as:
\begin{align} \label{Sdefn}
\tilde S: \mathcal{C} &\longrightarrow {\mathcal{H}}_{\hat{\mathcal{C}}}\\
\bigotimes_{x \in \mathbb{Z}^n} \vert {c}_{x} \rangle  &\mapsto \bigotimes_{x \in \mathbb{Z}^n} S(\vert {c}_{x} \rangle) \nonumber
\end{align}
and  extended to  ${\mathcal{H}}_{\mathcal{C}}$.  Here $S: W \longrightarrow \hat W$ is as in Theorem~\ref{thmS}, and with the choice of inner product on $\hat W = \bigotimes_{z \in \mathcal{N}} V_z$ as discussed at the beginning of Section~\ref{sqsbl}, $S$  is an inner product preserving isomorphism of Hilbert spaces $W$ and $\hat W$. This implies $\tilde S$ is an inner product preserving isomorphism of Hilbert spaces ${\mathcal{H}}_{\mathcal{C}}$ and ${\mathcal{H}}_{\hat{\mathcal{C}}}$.

\subsection{Quantum cellular automata and quantum lattice-gases}

Now we combine the foregoing results to prove our main theorem and characterize the class of QCAs that are QLGAs. 

For a QCA  defined by $({R},\mathcal{H}_{\mathcal{C}})$, with   neighborhood $\mathcal{N}$,  let $\mathcal{D}_{{x-y},x}$ be as defined in~\eqref{Projectiondef}, where  $x \in \mathbb{Z}^n$, and $y \in  \mathcal{N}$.   Assuming that $R$ obeys  the condition  in Theorem~\ref{thmS}$: \mathcal{A}_x = \begin{rm}{span }\end{rm}(\prod_{y \in  \mathcal{N}} \mathcal{D}_{{x-y},x})$,  we can construct a locally equivalent QCA $(R',\mathcal{H}_{\mathcal{C}'})$  (as defined in Section~\ref{loceqv}), with a quiescent symbol $\vert q'_0 \rangle$, whose  image  under  the isomorphism $S$ in  Theorem~\ref{thmS},  satisfies:
\begin{equation*}  
\vert \hat{q}'_0 \rangle = S(\vert q'_0 \rangle) \in \bigotimes_{z \in \mathcal{N}} V_z,~{\rm is~a~simple~tensor,~i.e.,}~\vert \hat{q}'_0 \rangle = \bigotimes_{z \in \mathcal{N}} \vert \hat{q}'_{0,z} \rangle{\rm,~where}~\vert \hat{q}'_{0,z} \rangle  \in V_z. 
\end{equation*}
Under the condition in Theorem~\ref{thmS},  therefore, we may as well consider a  QCA $(R,\mathcal{H}_{\mathcal{ C}})$  for which $\vert \hat{q}_0 \rangle$ is a simple tensor.   That is, up to local unitary equivalence:

\begin{thm}\label{thmStructure}
$R$ is the  global evolution of a QCA  on the Hilbert space of finite configurations $\mathcal{H}_{\mathcal{ C}}$, with   neighborhood $\mathcal{N}$,  and satisfies: 
 $\mathcal{A}_x = \begin{rm}{span }\end{rm}(\prod_{y \in  \mathcal{N}} \mathcal{D}_{{x-y},x})$,  if and only if:
\begin{enumerate}[label=(\roman{*})]
\item \label{thmitem1}  There exists  an  isomorphism of vector spaces $S$:
\begin{equation*}
S :  W \longrightarrow \bigotimes_{z \in \mathcal{N}} V_{z}
\end{equation*}
for some vector spaces $\{V_{z}\}_{z \in \mathcal{N}}$.   Under the isomorphism $S$, for each $y   \in \mathcal{N}$:
\begin{equation*}
\mathcal{D}_{{x-y},x} \cong  \begin{rm}{End}\end{rm}(V_y) \otimes \bigotimes_{z \in \mathcal{N}, z \neq y} \mathbb{I}_{V_z}
\end{equation*}
Furthermore, $\hat W =  \bigotimes_{z \in \mathcal{N}} V_{z}$ can be given  an inner product such that $S$ is an inner product preserving isomorphism of Hilbert spaces.
\item  \label{thmitem2}  $R$ is given by:
\begin{equation*}
R   =    {{\tilde S}^{-1}}    {\hat F} \sigma      \tilde S
\end{equation*}
where   $\sigma$ is as in~\eqref{tfdefin}, $\hat F$ is as in~\eqref{Fdefn} in terms of a unitary map $F$ on $\bigotimes_{z \in \mathcal{N}} V_{z}$, and  $\tilde S$ is as in~\eqref{Sdefn}.
\item \label{thmitem3} $F$, in the definition of $ \hat F$,  has  $\vert \hat{q}_0 \rangle$ as an invariant: $F \vert \hat{q}_0 \rangle =  \vert \hat{q}_0 \rangle$.
\end{enumerate}
\end{thm}

Before proving the theorem, we summarize its implications.
\begin{remarknonum}
Theorem~\ref{thmStructure} proves that QCA that satisfy $\mathcal{A}_x = \begin{rm}{span }\end{rm}(\prod_{y \in  \mathcal{N}} \mathcal{D}_{{x-y},x})$  are equivalent to QLGA, and vice versa in the following sense. Part~\ref{thmitem1} of Theorem~\ref{thmStructure} shows that there is an isomorphism $S$ from the cell algebra of the QCA to the states of the lattice vectors of the QLGA. Part~\ref{thmitem2} shows that applying $S$ to the Hilbert space of each cell enables $R$ to be realized as a product of collision and propagation  operators. Hence Theorem~\ref{thmStructure} identifies QCA satisfying $\mathcal{A}_x = \begin{rm}{span }\end{rm}(\prod_{y \in  \mathcal{N}} \mathcal{D}_{{x-y},x})$ with QLGA as defined in subsection~\ref{qlgadef}. 
\end{remarknonum}

\begin{proof}[Proof of Theorem~\ref{thmStructure}]
 We first prove that if $R$ is the global evolution of a QCA,  and satisfies: 
\begin{equation*}
\mathcal{A}_x = \begin{rm}{span }\end{rm}(\prod_{y \in  \mathcal{N}} \mathcal{D}_{{x-y},x})\text{,}
\end{equation*}
 then ~\ref{thmitem1}, ~\ref{thmitem2} and \ref{thmitem3} follow. 
 
 By Theorem~\ref{thmS}, the vector space isomorphism $S$ in~\ref{thmitem1} is immediate.  We can imbue $\hat W = S(W) = \bigotimes_{z \in \mathcal{N}} V_z$ with an inner product such that $\{S(\vert q \rangle) : q \in Q\}$ is an orthonormal set (as discussed in the beginning of Section~\ref{sqsbl}). By this choice of inner product,  $\hat W$ is a Hilbert space such that $S$ is an inner product preserving isomorphism of Hilbert spaces, hence proving~\ref{thmitem1}.   
 
 By corollary~\ref{corS}~\ref{cor3}, $\mathcal{R}_x = \begin{rm}{span }\end{rm}(\prod_{z \in  \mathcal{N}} \mathcal{D}_{x,x+z})$, where the definition of $\mathcal{D}_{x,x+z}$ is as in~\eqref{Projectiondef}. It follows that:
\begin{equation} \label{mseq}
 \tilde S \mathcal{R}_x {\tilde S}^{-1}   =  \begin{rm}{span }\end{rm}(\prod_{z \in  \mathcal{N}} \tilde S \mathcal{D}_{x,x+z} {\tilde S}^{-1})
\end{equation}
where $\tilde S$ is as in~\eqref{Sdefn}. We use  the indexing from~\eqref{chattrans},~\eqref{tensorindex}, and~\eqref{componentindex}. The isomorphisms from Theorem~\ref{thmS} give us:
\begin{align*}
\mathcal{A}_x &= {\tilde S}^{-1} \big(\bigotimes_{z \in \mathcal{N}} \begin{rm}{End}\end{rm}(V_{z,x})  \big) \tilde S \\
\tilde S \mathcal{D}_{x,x+z} {\tilde S}^{-1} &=  \begin{rm}{End}\end{rm}(V_{z,x+z}) \text{,}
\end{align*}
Using the above, we expand the LHS of~\eqref{mseq}:
\begin{equation*}
 \tilde S \mathcal{R}_x {\tilde S}^{-1}  = \tilde S R^\dag \mathcal{A}_x  R {\tilde S}^{-1}  = (\tilde S R^\dag {\tilde S}^{-1}) \big(\bigotimes_{z \in \mathcal{N}} \begin{rm}{End}\end{rm}(V_{z,x})  \big) (\tilde S R {\tilde S}^{-1})
\end{equation*}
and  the RHS:
\begin{equation*}
 \begin{rm}{span }\end{rm}(\prod_{z \in  \mathcal{N}} \tilde S \mathcal{D}_{x,x+z} {\tilde S}^{-1})  =  \begin{rm}{span }\end{rm}(\prod_{z \in  \mathcal{N}} \begin{rm}{End}\end{rm}(V_{z,x+z})) = \bigotimes_{z \in  \mathcal{N}} \begin{rm}{End}\end{rm}(V_{z,x+z})
\end{equation*}
and  obtain:
\begin{equation} \label{smseq}
(\tilde S R^\dag {\tilde S}^{-1}) \big(\bigotimes_{z \in \mathcal{N}} \begin{rm}{End}\end{rm}(V_{z,x})  \big) (\tilde S R {\tilde S}^{-1})    =  \bigotimes_{z \in  \mathcal{N}} \begin{rm}{End}\end{rm}(V_{z,x+z})
\end{equation}
Let us denote $\hat R = \tilde S R {\tilde S}^{-1}$. 
\begin{align*}
\hat R =   \tilde S R {\tilde S}^{-1}  :   {\mathcal{H}}_{\hat{\mathcal{C}}} \longrightarrow  {\mathcal{H}}_{\hat{\mathcal{C}}}.
\end{align*}
$\hat R$ is clearly unitary as  $R$ is unitary and $\tilde S$  is inner product preserving ($\tilde S$  is composed of inner product preserving operator $S$ acting on every cell),  and is translation-invariant since $R$ and $\tilde S$ are translation-invariant.  

Fix $x \in \mathbb{Z}^n$. Choose a  basis, $B_{V_x}$,   of   ${\hat W}_x = \bigotimes_{z \in  \mathcal{N}} V_{z,x}$:
\begin{equation} \label{vzxbasis}
B_{V_x} = \{\vert b_x \rangle \}  \subset \bigotimes_{z \in  \mathcal{N}} V_{z,x}
\end{equation}
For instance, we can choose $B_{V_x}$ to be the product basis of ${\hat W}_x$ of the form~\eqref{Bnu0def}, letting : $B_{V_x} = \{ \bigotimes_{z \in \mathcal{N}} \vert \textbf{k}_x(z) \rangle : \textbf{k}_x \in \prod_{z \in \mathcal{N}} \{1,..,d_{V_{z}}\} \}$. 
The set $\{ \vert b_x \rangle \otimes \langle b'_x \vert : \vert b_x \rangle, \vert b'_x \rangle \in B_{V_x}\}$  is a basis of $\begin{rm}{End}\end{rm}({\hat W}_x)  = \bigotimes_{z \in  \mathcal{N}} \begin{rm}{End}\end{rm}(V_{z,x})$. Choose such a basis element of $\begin{rm}{End}\end{rm}({\hat W}_x) $ and create a  basis operator local on cell $x$: 
\begin{equation} \label{rhatdecomp0}
{\hat A}_x =(\vert b_x \rangle \otimes \langle b'_x \vert) \otimes \mathbb{I}_{\overline{\{x\}}}
\end{equation}
where $\mathbb{I}_{\overline{\{x\}}}$ is the identity on co-$\{x\}$ space (defined in Definition~\ref{coDdef} for ${\mathcal{H}}_{\mathcal{C}}$, and the definition applies equally   for ${\mathcal{H}}_{\hat{\mathcal{C}}}$ by replacing $\mathcal{C}$ with $\hat{\mathcal{C}}$). $\mathbb{I}_{\overline{\{x\}}}$ is included as part of the local operator to help make the argument clear .
\begin{equation*}
 {\hat R}^\dag {\hat A}_x \hat R     =  {\hat R}^\dag  \big((\vert b_x \rangle \otimes \langle b'_x \vert ) \otimes \mathbb{I}_{\overline{\{x\}}}\big)  \hat R  
\end{equation*}
By~\eqref{smseq}, $ {\hat R}^\dag {\hat A}_x \hat R$  is a rank one element of $\bigotimes_{z \in  \mathcal{N}} \begin{rm}{End}\end{rm}(V_{z,x+z})$, whereas, by definition,  ${\hat A}_x$ is a rank-one element of $\begin{rm}{End}\end{rm}({\hat W}_x) =  \bigotimes_{z \in  \mathcal{N}} \begin{rm}{End}\end{rm}(V_{z,x})$. Since  ${\hat R}^\dag {\hat A}_x \hat R$ is  an element of $\bigotimes_{z \in  \mathcal{N}} \begin{rm}{End}\end{rm}(V_{z,x+z})$ for all basis elements ${\hat A}_x \in \begin{rm}{End}\end{rm}({\hat W}_x)$ of the form chosen,    this implies that we can write $  {\hat R}^\dag {\hat A}_x \hat R$ as follows:
\begin{equation} \label{rhatdecomp}
  {\hat R}^\dag {\hat A}_x \hat R = \sigma^\dag_x   \big((\vert w_x \rangle \otimes \langle w'_x \vert) \otimes \mathbb{I}_{\overline{\{x\}}} \big) \sigma_x.
\end{equation}
for some $\vert w_x \rangle, \vert w'_x \rangle \in {\hat W}_x = \bigotimes_{z \in  \mathcal{N}} V_{z,x}$, and  an appropriate permutation operator $\sigma_x$.     $\sigma_x$  maps  the components of the  neighborhood elements of cell $x$, $ \bigotimes_{z \in  \mathcal{N}} V_{z,x+z}$,     to the corresponding components  of cell $x$, $\bigotimes_{z \in  \mathcal{N}} V_{z,x}$.     

We label the combined cell and  component indices together by  the pair $(t,z) \in \mathbb{Z}^n \times \mathcal{N}$.  It is necessary and sufficient that  $\sigma_x$  satisfy the following conditions:
\begin{enumerate}[label=(\alph{*})] 
\item \label{sigmaxprop1} $\sigma_x$ permutes  the $z$-th  component of a  cell  to $z$-th component of another cell (where $z \in \mathcal{N}$). That is, $\sigma_x$ permutes the cell indices  $t \in \mathbb{Z}^n$, leaving the component indices  $ z \in \mathcal{N}$ unchanged. 
\item \label{sigmaxprop2} For the cell index $t = x$, $\sigma_x$    maps the $z$-th component of the neighborhood  cell $x+z$ to the $z$-th component of cell $x$ (where $z \in \mathcal{N}$).
\end{enumerate}
To satisfy the requirement in condition~\ref{sigmaxprop1}  that $\sigma_x$ only permute  the cell indices,  keeping the component indices intact,  we use  a cell-indexing  function, $s_x: \mathbb{Z}^n \times \mathcal{N} \longrightarrow \mathbb{Z}^n$, in the definition of  $\sigma_x$.  As every element  $\vert {\hat{c}}\rangle \in \hat{\mathcal{C}}$ (Definition~\ref{Chatsofc})  is of the form $\vert {\hat{c}}\rangle = \bigotimes_{t \in \mathbb{Z}^n} \bigotimes_{z \in \mathcal{N}} \vert \textbf{k}_t(z) \rangle$~\eqref{chattrans}, we can define $\sigma_x$ as follows: 
  \begin{equation*} 
 \left . \begin{array} {cccc}
 \sigma_x : &\hat{\mathcal{C}}&\longrightarrow  & \hat{\mathcal{C}} \\
  &\bigotimes_{t \in \mathbb{Z}^n} \bigotimes_{z \in \mathcal{N}} \vert \textbf{k}_t(z) \rangle  &\mapsto  &\bigotimes_{t \in \mathbb{Z}^n} \bigotimes_{z \in \mathcal{N}} \vert  \textbf{k}_{s_x(t,z)}(z) \rangle
    \end{array}
  \right .
\end{equation*} 
The rest of the requirements in conditions~\ref{sigmaxprop1} and ~\ref{sigmaxprop2} above can be recast as restrictions on $s_x$. The remaining requirement in condition~\ref{sigmaxprop1} is that  $\sigma_x$ is a permutation. This implies that  when the component index  of the domain of $s_x$ is fixed to be some $z \in \mathcal{N}$,  $s_x$   is a bijection of the cell indices $\mathbb{Z}^n$, i.e, a permutation of cell indices.
\begin{equation*}
s_x\mid_{\mathbb{Z}^n \times z}: \mathbb{Z}^n \times z \longleftrightarrow \mathbb{Z}^n\\
\end{equation*}
Condition~\ref{sigmaxprop2} implies that   when  the cell index $t \in \mathbb{Z}^n$ of the domain of $s_x$ is  fixed to be $x$, $s_x$  is as follows:
\begin{align*}
s_x\mid_{x \times \mathcal{N}}: x \times \mathcal{N} &\longrightarrow \mathbb{Z}^n \\
(x,z) &\mapsto x+z
\end{align*} 
$\sigma_x$ thus defined on $\hat{\mathcal{C}}$ under the constraints~\ref{sigmaxprop1} and~\ref{sigmaxprop2},   is extended to  ${\mathcal{H}}_{\hat{\mathcal{C}}}$.
 
Now consider the implication of combining~\eqref{rhatdecomp0} and~\eqref{rhatdecomp}:
\begin{equation}  \label{Axbasisform}
 {\hat R}^\dag  \big((\vert b_x \rangle \otimes \langle b'_x \vert ) \otimes \mathbb{I}_{\overline{\{x\}}}\big)  \hat R   = \sigma^\dag_x  \big((\vert w_x \rangle \otimes \langle w'_x \vert) \otimes \mathbb{I}_{\overline{\{x\}}}\big) \sigma_x.
\end{equation}
We have the above equality for every $\vert b_x \rangle \otimes \langle b'_x \vert$ and for the corresponding $\vert w_x \rangle \otimes \langle w'_x \vert$. So we define a linear map $T_x$ on ${\hat W}_x =  \bigotimes_{z \in  \mathcal{N}} V_{z,x}$, by defining it on the  basis $B_{V_x}$~\eqref{vzxbasis}:
\begin{align*}
T_x: B_{V_x} &\longrightarrow {\hat W}_x \\
\vert b_x \rangle &\mapsto \vert w_x \rangle
\end{align*}
where $\vert b_x \rangle$ and $\vert w_x \rangle$ are  as in~\eqref{Axbasisform}. $T_x$ is 
 extended  to ${\hat W}_x$ by linearity. Then $T_x$ is unitary since $\hat R$ is unitary. By symmetry, we can write~\eqref{Axbasisform} as:
\begin{equation*} 
 {\hat R}^\dag  \big((\vert b_x \rangle \otimes \langle b'_x \vert ) \otimes \mathbb{I}_{\overline{\{x\}}}\big)  \hat R    = \sigma^\dag_x  \big((T_x\vert b_x \rangle \otimes \langle b'_x \vert T^{\dag}_x) \otimes \mathbb{I}_{\overline{\{x\}}}\big) \sigma_x
\end{equation*}
This is true for every $x\in \mathbb{Z}^n$. In addition, $\hat R = \tilde S R {\tilde S}^{-1}$ is  translation-invariant, which implies that all the maps $T_x$ can be taken to be the same unitary map $T$ on $\hat W = \bigotimes_{z \in  \mathcal{N}} V_{z}$, and that the map $\sigma_x$ for each $x$ is translation-invariant, i.e., each $\sigma_x$ is the  map $\sigma$ defined in~\eqref{tfdefin}.  Now define $F: \hat W \mapsto \hat W$, to be  $F := T^\dag$. Then   $\hat R =  \tilde S R {\tilde S}^{-1}$ is a  composition of  the propagation operator $\sigma$ followed by the collision operator  $\hat F$: 
\begin{equation*}
 \tilde S R {{\tilde S}^{-1}}   =          {\hat F}   \sigma
\end{equation*}
where  $\sigma$  is as defined in~\eqref{tfdefin}, and  ${\hat F}$ is a unitary map on ${\mathcal{H}}_{\hat{\mathcal{C}}}$  as defined in~\eqref{Fdefn} in terms of $F$. Thus we get part~\ref{thmitem2} of the theorem. By Lemma~\ref{Fdefm} we get  part \ref{thmitem3} of the theorem: $F$ in the definition of $ \hat F$ has  $\vert \hat{q}_0 \rangle$ as an invariant, i.e.,  $F \vert \hat{q}_0 \rangle =   \vert \hat{q}_0 \rangle$. 

Next we prove that if~\ref{thmitem1}, ~\ref{thmitem2}  and \ref{thmitem3} hold then $R$ is the global evolution of a QCA with neighborhood $\mathcal{N}$, and that it satisfies the condition $\mathcal{A}_x = \begin{rm}{span }\end{rm}(\prod_{y \in  \mathcal{N}} \mathcal{D}_{{x-y},x})$. Unitarity and translation-invariance are obvious from the definition.  To prove  causality, let us assume $A_x$ is a local operator on cell  $x$. We consider the  action of ${ R}^\dag  A_x {R}$ on basis elements of ${\mathcal{H}}_{\mathcal{C}}$. Both $\tilde S$ and $\hat F$ act locally on cells by $S$ and $F$ respectively. The conjugation of ${\mathcal{A}}_x$ by $S$ is an isomorphism of ${\mathcal{A}}_x$ to ${\hat {\mathcal{A}}}_x =  \bigotimes_{z \in \mathcal{N}} \begin{rm}{End}\end{rm}(V_{z,x})$ (the counterpart of $\mathcal{A}_x$),  and the conjugation of ${\hat {\mathcal{A}}}_x$ by $F$  is an  isomorphism of  ${\hat {\mathcal{A}}}_x$ to itself.  Observe also that   $F \vert \hat{q}_0 \rangle = \vert \hat{q}_0 \rangle$. Then we may as well simply look at the conjugation of ${\hat {\mathcal{A}}}_x$ by $\sigma$, and  only need to show that propagation is causal. Therefore we can simply consider the action of ${\sigma}^{-1} {\hat {\mathcal{A}}}_x  \sigma $ on basis elements of ${\mathcal{H}}_{\hat{\mathcal{C}}}$.  We can write  such a basis element as:
\begin{equation*}
\vert \hat c \rangle = \bigotimes_{y \in\mathbb{Z}^n}   \big( \bigotimes_{z \in \mathcal{N}} \vert \textbf{k}_y(z) \rangle \big)
\end{equation*}
Then:
\begin{equation*}
\sigma \vert \hat c \rangle =  \bigotimes_{y \in\mathbb{Z}^n} \big( \bigotimes_{z \in \mathcal{N}} \vert \textbf{k}_{y+z}(z) \rangle \big)
\end{equation*}
There is no loss of generality (by linearity) in assuming that $\hat {A}_x \in {\hat {\mathcal{A}}}_x$ is a basis element:  
\begin{equation*}\hat {A}_x =\bigg(\big( \bigotimes_{r \in \mathcal{N}} \vert a_x(r) \rangle \big)\otimes  \big( \bigotimes_{r' \in \mathcal{N}} \langle a_x'(r') \vert \big) \bigg)\otimes \mathbb{I}_{\overline{\{x\}}}
\end{equation*}
where $\vert a_x(r) \rangle \in V_r$, and  $\vert a_x'(r') \rangle \in  V_{r'}$,  and where $\mathbb{I}_{\overline{\{x\}}}$ is the identity on co-$\{x\}$ space, as in Definition~\ref{coDdef}.
Denote by $p_x  = \prod_{z \in \mathcal{N}} \langle a_x'(z) \vert \textbf{k}_{x+z}(z) \rangle$. Then this implies:
\begin{equation*} 
{\sigma}^{-1} \hat {A}_x  \sigma \vert \hat c \rangle =    p_x \bigotimes_{y = x+z : z \in \mathcal{N}} \big(\underbrace{\vert a_x(z) \rangle}_{z' = z}  \otimes \bigotimes_{z' \in \mathcal{N}\setminus \{z\}} \vert \textbf{k}_{y}(z') \rangle \big)   \otimes \big(  \bigotimes_{y \in\mathbb{Z}^n\setminus \{x+z : z \in \mathcal{N}\}} \bigotimes_{z' \in \mathcal{N}} \vert \textbf{k}_{y}(z') \rangle \big)
\end{equation*}
The above  shows that ${\sigma}^{-1} \hat {A}_x  \sigma$  is local upon $\mathcal{N}_x$, hence by Theorem~\ref{strucreva} (Structural Reversibility),  $\hat R$ is causal relative to neighborhood $\mathcal{N}$. This in turn implies,  by  cell-wise conjugation with  local isomorphisms $S$ and $F$,  that $R$ is causal relative to  neighborhood $\mathcal{N}$.

Next we show that  $\mathcal{A}_x = \begin{rm}{span }\end{rm}(\prod_{y \in  \mathcal{N}} \mathcal{D}_{{x-y},x})$.  By the same reasoning used in showing causality, we can consider the following quantities. Let ${\hat {\mathcal{B}}}_z =  \sigma^{-1}  {\hat {\mathcal{A}}}_z \sigma$ (the counterpart of $\mathcal{R}_z$).  Let ${\hat {\mathcal{D}}}_{x-y,x} = {\hat {\mathcal{B}}}_{x-y} \cap  {\hat {\mathcal{A}}}_x$ (the counterpart of $\mathcal{D}_{x-y, x}$), for $y   \in \mathcal{N}$. Then  it is clear by the definition of $\sigma$, that:
\begin{equation*}
{\hat {\mathcal{D}}}_{x-y,x} = {\hat {\mathcal{B}}}_{x-y}   \cap  {\hat {\mathcal{A}}}_x=  \begin{rm}{End}\end{rm}(V_{y,x}) \otimes \bigotimes_{z \in \mathcal{N}, z \neq y} \mathbb{I}_{V_{z,x}}
\end{equation*}
for each $y   \in \mathcal{N}$. This implies:
\begin{equation*}
{\hat {\mathcal{A}}}_x=  \begin{rm}{span }\end{rm}(\prod_{y \in  \mathcal{N}} {\hat {\mathcal{D}}}_{x-y,x})
\end{equation*}
This further implies,   by  cell-wise conjugation with  local isomorphisms $S$ and $F$, that:  $\mathcal{A}_x = \begin{rm}{span }\end{rm}(\prod_{y \in  \mathcal{N}} \mathcal{D}_{{x-y},x})$.
\end{proof}

The significance of this theorem is that it connects a local condition to the global structure of the QCA. We state the theorem in terms of propagation and collision operators which are the quantum analogues of the substeps of the evolution for a classical lattice-gas. There are, however, a few  equivalent ways of expressing the global evolution $R$ under the conditions of Theorem~\ref{thmStructure}.

\begin{cor}  \label{rcor} The global evolution in Theorem~\ref{thmStructure} can  be given as:
\begin{equation*}
R   =      \tilde T   \sigma  \tilde S
\end{equation*}
for some map $\tilde T$:
\begin{equation*}
\tilde T:  {\mathcal{H}}_{\hat{\mathcal{C}}} \longrightarrow {\mathcal{H}}_{\mathcal{C}} 
\end{equation*}
\end{cor}

\begin{proof}
Let:
\begin{equation*}
\tilde T:  {\mathcal{H}}_{\hat{\mathcal{C}}} \longrightarrow {\mathcal{H}}_{\mathcal{C}} 
\end{equation*}
on  $\hat{\mathcal{C}}$ as:
\begin{align*}
\tilde T:\hat{\mathcal{C}} &\longrightarrow {\mathcal{H}}_{\mathcal{C}}\\
\bigotimes_{x \in \mathbb{Z}^n} \vert {c}_{x} \rangle  &\mapsto \bigotimes_{x \in \mathbb{Z}^n} S^{-1} F  (\vert {c}_{x} \rangle)
\end{align*}
$S$, and $F$ are as in Theorem~\ref{thmStructure}.  $\tilde T$ is extended      to  ${\mathcal{H}}_{\hat{\mathcal{C}}}$. 

 By Theorem~\ref{thmStructure}:
\begin{equation*}
R   =      \tilde T   \sigma  \tilde S
\end{equation*}
\end{proof}
This is called the two-layered \textit{brick structure}~\cite{sw:rvqca}. However, instead of working with the space ${\mathcal{H}}_{\mathcal{C}}$, we can look at the QCA global evolution as it acts on  the space ${\mathcal{H}}_{\hat{\mathcal{C}}}$. This gives us another equivalent form for the global evolution $R$.

\begin{cor}  \label{rtildecor} The global evolution in Theorem~\ref{thmStructure} can equivalently be given by a unitary operator $\hat R$ on ${\mathcal{H}}_{\hat{\mathcal{C}}}$:
\begin{equation*}
 \hat R   =    \hat F  \sigma 
\end{equation*}
\end{cor}

\begin{proof}
Let:  
\begin{align*}
 \hat R   &:   {\mathcal{H}}_{\hat{\mathcal{C}}} \longrightarrow {\mathcal{H}}_{\hat{\mathcal{C}}} \\
 \hat R   &=    \tilde S R {{\tilde S}^{-1}}
\end{align*}
 By Theorem~\ref{thmStructure}:
\begin{align*}
 \hat R   &=    \hat F  \sigma 
\end{align*}
\end{proof}

\begin{remarknonum}
 Another description of the QLGA is obtained under the unitary isomorphism:
 \begin{equation*}
 \sigma:  {\mathcal{H}}_{\hat{\mathcal{C}}} \rightarrow {\mathcal{H}}_{\hat{\mathcal{C}}}
 \end{equation*}
Under this isomorphism, the global evolution becomes:
 \begin{equation*} 
\check R   =      \sigma \hat R  \sigma^{-1}  = \sigma  {\hat F}
\end{equation*}
\end{remarknonum}

Thus the global evolution $R$ is equivalent to one given by another global evolution operator $\hat R$, with $\tilde S$  and  ${{\tilde S}^{-1}}$ providing the change of basis between  ${\mathcal{H}}_{\mathcal{C}}$ and  ${\mathcal{H}}_{\hat{\mathcal{C}}}$. The  equivalent evolution $\hat R$ happens in  two stages:  the \textit{propagation} (or \textit{advection}) stage given by  $\sigma$, followed by the the \textit{collision} (or \textit{scattering}) stage given by the  map $F$ on $\bigotimes_{z \in \mathcal{N}} V_z$. This completes the proof of equivalence of QCA satisfying the local condition of Theorem~\ref{thmStructure} to \textit{Quantum Lattice-Gas Automata} (QLGA). Theorem~\ref{thmStructure} therefore characterizes the QLGA as a subclass of QCA. It is clear that the evolution given by Theorem~\ref{thmStructure} has as special cases both the classical lattice-gas defined in the introduction and the lattice-gas models previously investigated~\cite{bib:meyer1,bib:meyer2,bib:meyer3,bib:meyer4,bib:meyer5,bib:meyer6,Bog1,Bog2,Bog3}. In the following section we will give some examples of QCA which are {\em not} lattice-gases and show how they violate the conditions of Theorem~\ref{thmStructure}. We also compute the matrices $\mathcal{D}_{x,y}$ explicitly for the simplest quantum lattice-gas and show that the condition is indeed satisfied.

%%%%%%%%%%%%%%%%%%%%%%%%%%%%%%%%%%%%%%%%%%%%%%%%%%%%%%%%%%%%%%%%%%%%%%%%%%%%%%
\section{Examples}\label{ex}

In this section we consider examples of three subclasses of QCA. In the first case we examine the simplest quantum lattice-gas, defined as an example in~\cite{bib:feynmanhibbs} and also obtained from consideration of the absence of scalar QCA in~\cite{bib:meyer2} and investigated in~\cite{bib:meyer2,bib:meyer3,bib:meyer4,bib:meyer5}. Here it is instructive to see the algebras that appear in our local condition. Secondly, we investigate the counterexample of~\cite{anw:odqca} to show that our condition is indeed not satisfied in this case. Finally, we give an example in which quantization of a classical automata changes the neighborhood size, and results in a QCA which is also not a QLGA. In this case we show how adding a propagation step to the original QCA allows our condition to be satisfied and hence results in a QLGA.

\subsection{The simplest QLGA}

The simplest QLGA was originally investigated by Meyer as an extension of a single particle partitioned QCA~\cite{bib:meyer2}. In this model each site has two neighbors, and  each site has two lattice vectors that point to these neighbors. At most one particle may be present on each vector so the site Hilbert space $W$ is isomorphic to $\mathbb{C}^4\simeq\mathbb{C}^2\otimes\mathbb{C}^2$. The collision operator $F$ is particle number conserving - meaning that it is block diagonal with two blocks of dimension $1$ and one block of dimension $2$. We may choose the basis $\{\vert00\rangle,\vert01\rangle,\vert10\rangle,\vert11\rangle\}$ for $W$ in which the collision operator $F$ may be written:
\begin{equation*}
F=\begin{pmatrix}
1&0&0&0\\
0&ie^{i\alpha}\sin\theta&e^{i\alpha}\cos\theta&0\\
0&e^{i\alpha}\cos\theta&ie^{i\alpha}\sin\theta&0\\
0&0&0&e^{i\beta}\\
\end{pmatrix}.
\end{equation*}

The lattice is $\mathbb{Z}$ and the neighborhood is $\mathcal{N}=\{-1,+1\}$. Because we are beginning from a QLGA dynamics composed of propagation and collision we may take the map $S$ to be the identity. Then
\begin{equation*}
W=\hat W = \mathbb{C}^2\otimes\mathbb{C}^2
\end{equation*} 
and the two tensor factors are $V_{-1}\simeq \mathbb{C}^2$ and $V_{1}\simeq \mathbb{C}^2$. From the Heisenberg point of view of the dynamics of the algebras $\mathcal{A}_x$, which are localized on a lattice-gas site and hence act on $\mathbb{C}^4$, the collision step is simply some isomorphism of the algebra. We need therefore only to consider the action of the propagation operator to determine the algebras $\mathcal{R}_x$ and $\mathcal{D}_{x-y,x}$. The action of $\sigma$ on an operator localized on the site may be written:
\begin{equation*}
\sigma^{-1} \mathbb{I}\otimes B\otimes\mathbb{I}\otimes \mathbb{I}\otimes A\otimes \mathbb{I}\sigma = \mathbb{I}\otimes\mathbb{I}\otimes A\otimes B \otimes  \mathbb{I}\otimes \mathbb{I}
\end{equation*}
where the first two tensor factors (reading from left to right) are the left lattice-gas site $x-1$, the middle two tensor factors are the middle lattice-gas site $x$ and the right two tensor factors are the right lattice-gas site $x+1$. From this it is straightforward to write:
\begin{align*}
\mathcal{D}_{x-1,x} &= \mathbb{I}\otimes \text{End}(\mathbb{C}^2)\\
\mathcal{D}_{x+1,x} &= \text{End}(\mathbb{C}^2)\otimes\mathbb{I}\\
\end{align*}
Evidently the span of the product of these algebras is $\text{End}(\mathbb{C}^4)$ and hence this model satisfies the condition to be a quantum lattice-gas, as it must.

\subsection{QCA that are not QLGA}

Consider a one-dimensional  classical  CA as in Figure~\ref{qlgafig} in which a cell is $3$ bits. Then a cell $x$ is given by a triple $(a_x,b_x,c_x) \in (\mathbb{Z}/(2))^3$. The neighborhood is $\mathcal{N} = \{-1, 0, 1\}$,   and the update rule $\mu$:
\begin{align*}
\mu: a_{x} &\mapsto  a_{x} + b_{x-1} \\
 b_{x}  &\mapsto  b_x  \\
 c_{x}   &\mapsto  c_{x} + b_{x+1} 
\end{align*}

Quantizing this CA, we get a QCA. That means,  we  take the rule above to describe the evolution of a cell state as a function of neighborhood state.  A cell $x$ of the QCA consists of $3$ qubits, with basis states:  $\vert a_x \rangle \otimes \vert b_x \rangle \otimes \vert c_x \rangle$, where $a_x , b_x ,  c_x  \in\mathbb{Z}/(2)$. The quiescent state is $\vert \hat{q}_0 \rangle = \vert 0 \rangle \otimes \vert 0\rangle \otimes \vert 0 \rangle$. The neighborhood is still  $\mathcal{N} = \{ -1,0, 1\}$.

 \begin{figure}[h] 
\includegraphics{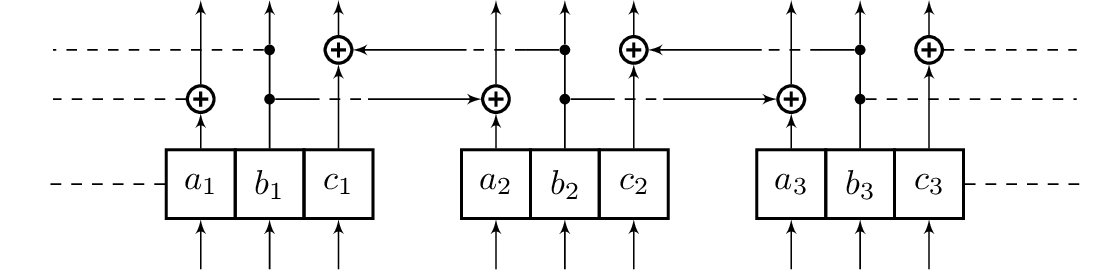}
\caption{One-dimensional CA which is a not a QLGA after quantization}
\label{qlgafig}
\end{figure}

Define the global evolution $R$ on ${\mathcal{H}}_{{\mathcal{C}}}$ through its action  by a controlled-NOT operation on the neighborhood cells as follows.
\begin{align*}
R: \vert a_{x} \rangle &\mapsto \vert a_{x} + b_{x-1} \rangle \\
\vert b_{x} \rangle &\mapsto \vert b_x \rangle \\
\vert c_{x}   \rangle &\mapsto \vert c_{x} + b_{x+1} \rangle
\end{align*}
Let us write the element of Hadamard basis as: $\vert \chi_a \rangle  =  \frac{1}{\sqrt{2}}(\vert 0 \rangle +  (-1)^{a}  \vert 1 \rangle)$ for  $a \in \mathbb{Z}/(2)$. Note that $\vert \chi_0 \rangle = \vert + \rangle$ and $\vert \chi_1 \rangle  = \vert - \rangle$. We can equivalently look at the qubits in basis states: $\vert \chi_{a_x} \rangle \otimes \vert b_x \rangle \otimes \vert \chi_{c_x} \rangle$, where $a_x , b_x ,  c_x  \in \mathbb{Z}/(2)$. Since  $\vert + \rangle$ and $\vert - \rangle$  are eigenvectors of the shift operation with eigenvalues $1$ and $-1$ respectively, $R$ can be given as:
\begin{align*}
R: \vert \chi_{a_x}\rangle &\mapsto \vert \chi_{a_x} \rangle \\
\vert b_{x} \rangle &\mapsto {(-1)}^{b_{x}(a_{x+1} +   c_{x-1})} \vert b_{x} \rangle \\
\vert  \chi_{c_x} \rangle &\mapsto \vert \chi_{c_x}  \rangle
\end{align*}
We now compute the  subalgebras $\mathcal{D}_{{x-y},x} =   \mathcal{R}_{x-y} \cap \mathcal{A}_x$, where $y\in \mathcal{N}$.  $\mathcal{D}_{x,x} =   \mathcal{R}_x \cap \mathcal{A}_x = \text{span}(\{\vert a_x \rangle \langle a_x \vert \otimes \vert b_x \rangle \langle b_x \vert \otimes \vert c_x \rangle \langle c_x \vert \})$, where $ \vert a_x \rangle, \vert c_x \rangle \in \{\vert + \rangle, \vert - \rangle\}$, and $ \vert b_x \rangle \in \{\vert 0 \rangle, \vert 1 \rangle\}$. $\mathcal{D}_{x-1,x} =   \mathcal{R}_{x-1} \cap \mathcal{A}_x = \mathbb{C} \mathbb{I} \otimes \mathbb{I} \otimes  \mathbb{I}$. $\mathcal{D}_{x+1,x} =   \mathcal{R}_{x+1} \cap \mathcal{A}_x = \mathbb{C} \mathbb{I} \otimes \mathbb{I} \otimes  \mathbb{I}$. Then it is clear that  dim $\begin{rm}{span }\end{rm}(\prod_{y \in  \{-1,0,1\}} \mathcal{D}_{{x-y},x}) = 8$, while dim $ \mathcal{A}_x  = 64$.  Theorem~\ref{thmStructure} then implies that  this QCA is not a QLGA.

The $2$-dimensional version of the above example is referred to as the  Kari CA in Arrighi, Nesme, and Werner~\cite{anw:odqca}. This is illustrated in  Figure~\ref{qlgafig2}. The lattice $\mathcal{L} = \mathbb{Z}\times  \mathbb{Z}$. A cell of this classical CA consists of $9$ bits. The neighborhood is $\mathcal{N} = \{ -1,0, 1\}\times \{ -1,0, 1\}$. The center bit of the center cell acts by controlled-NOT on the appropriate bit of the neighboring cell as shown.

 \begin{figure}[h] 
\includegraphics{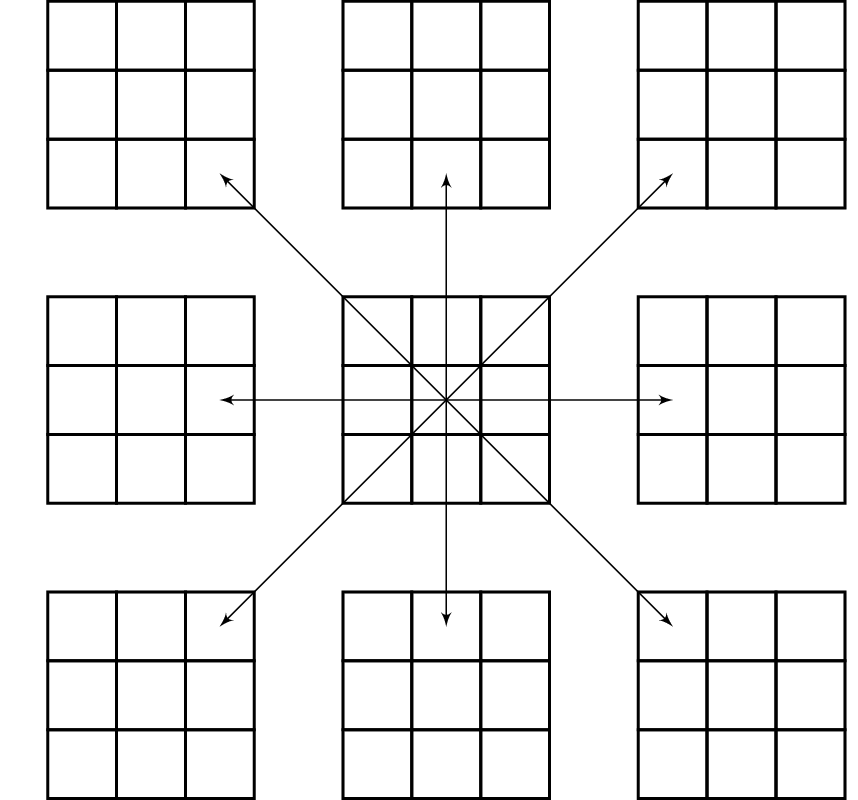}
\caption{Two-dimensional Kari CA which is a not a QLGA after quantization}
\label{qlgafig2}
\end{figure}

In the quantized version, the cells of the QCA are $9$-qubits each. Enumerate the qubits of  cell $x$ by $(i_x,j_x) \in \{ -1,0, 1\} \times \{ -1,0, 1\}$ (not to be confused with the neighborhood coordinates on the lattice). $(i_x,j_x) = (0,0)$ corresponds to the center qubit of a cell. The quiescent state is $\vert \hat{q}_0 \rangle = \bigotimes_{(i_x,j_x) \in \{ -1,0, 1\}\times \{ -1,0, 1\}} \vert 0 \rangle$. The center qubit of the center cell acts by controlled-NOT on the appropriate qubit  of the neighboring cell as for the classical CA. This action is exactly as in the $1$-dimensional case.  The neighborhood is the same as for the classical CA,  $\mathcal{N} = \{ -1,0, 1\}\times \{ -1,0, 1\}$.

 We compute, as in the $1$-dimensional case, the subalgebras $\mathcal{D}_{{x-y},x} =   \mathcal{R}_{x-y} \cap \mathcal{A}_x$, where $y\in \mathcal{N}$. $\mathcal{D}_{x,x} =   \mathcal{R}_x \cap \mathcal{A}_x = \text{span}(\prod_{(i_x,j_x) \in \{ -1,0, 1\}\times \{ -1,0, 1\}}\vert a_{i_xj_x} \rangle \langle a_{i_xj_x} \vert )$, where $\vert a_{00} \rangle \in \{\vert 0 \rangle, \vert 1 \rangle\}$, and the rest of $ \vert a_{i_xj_x} \rangle \in \{\vert + \rangle, \vert - \rangle\}$, for $(i_x,j_x) \neq (0,0)$. For all $y\in \mathcal{N}\setminus (0,0)$, $\mathcal{D}_{{x-y},x} =   \mathcal{R}_{x-y} \cap \mathcal{A}_x = \mathbb{C}  \bigotimes_{k \in   \{ -1,0, 1\}\times \{ -1,0, 1\}} \mathbb{I}$. Then it is clear that  dim $\begin{rm}{span }\end{rm}(\prod_{y \in  \mathcal{N}} \mathcal{D}_{{x-y},x}) = 2^9$, while dim $ \mathcal{A}_x  = 4^9$.  Theorem~\ref{thmStructure} then implies that  this QCA is also not a QLGA.

%%%%%%%%%%%%%%%%%%%%%%%%%%%%%%%%%%%%%%%%%%%%%%%%%%%%%%%%%%%%%%%%%%%%%%%%%%%%%%
\subsection{Examples of QCA that show effect of quantization on neighborhood}\label{nbhdchange}

Consider a one-dimensional classical  CA as in Figure~\ref{qlgafig3} in which a cell is $2$ bits. We consider the scheme in which  a cell $x$ is given by a pair $(a_x,b_x)$. The neighborhood is $\mathcal{N} = \{ 0, 1\}$,   and the update rule $\mu$:
\begin{align*}
\mu: a_{x} &\mapsto  a_{x}  \\
 b_{x}  &\mapsto  b_x +  a_{x} + a_{x+1}
\end{align*}
Let us quantize  this CA to get a QCA.   A cell $x$ of the QCA consists of $2$ qubits:  $\vert a_x \rangle \otimes \vert b_x \rangle$, where $a_x , b_x  \in\mathbb{Z}/(2)$. The quiescent state is $\vert \hat{q}_0 \rangle = \vert 0 \rangle \otimes \vert 0\rangle$. This automata is closely related to the ``Toffoli CA" in~\cite{anw:odqca}.

Define the global evolution $R$ on ${\mathcal{H}}_{{\mathcal{C}}}$ through its action  by a controlled-NOT operation on the basis of its neighborhood.
\begin{align*}
R: \vert a_{x} \rangle &\mapsto \vert a_x \rangle \\
\vert b_{x} \rangle &\mapsto \vert b_{x} + a_{x} + a_{x+1} \rangle 
\end{align*}
The neighborhood looks like   $\mathcal{N}_{+} = \{0, 1\}$ as expected.

 \begin{figure}[h] 
\centering 
\includegraphics{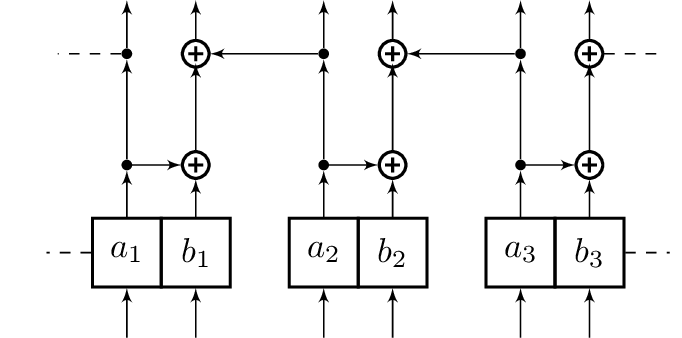}
\caption{One-dimensional CA whose neighborhood changes after quantization}
\label{qlgafig3}
\end{figure}

But we can equivalently look at the qubits in state: $ \vert a_x \rangle \otimes \vert \chi_{b_x} \rangle$, where $a_x , b_x   \in \mathbb{Z}/(2)$, and $\vert \chi_{b_x} \rangle$ is as defined in the examples above: $\vert \chi_{b_x} \rangle  =  \frac{1}{\sqrt{2}}(\vert 0 \rangle +  (-1)^{b_x}  \vert 1 \rangle)$ for  $b_x \in \mathbb{Z}/(2)$.  Since $\vert \chi_{b_x}  \rangle$ is an eigenvector  for the  shift operation, $R$ can be given as:
\begin{align*}
R:\quad\quad \vert a_x \rangle &\mapsto {(-1)}^{a_{x}(b_{x-1} +   b_{x})} \vert a_{x} \rangle \\
\vert \chi_{b_x} \rangle &\mapsto \vert \chi_{b_x}  \rangle
\end{align*}
Now the neighborhood looks like   $\mathcal{N}_{-} = \{-1,0\}$. In fact, the neighborhood is $\mathcal{N} = \mathcal{N}_{-} \cup \mathcal{N}_{+} = \{-1,0,1\}$. This is a consequence of the well-known fact that the control and target of a controlled-NOT operation in the basis $\vert + \rangle, \vert - \rangle$ are reversed as compared with the logical basis. Hence, while controlled-NOT operations can only propagate classical information in one direction (from control to target), they may propagate quantum information in both directions.

As before, we  compute the  subalgebras $\mathcal{D}_{{x-y},x} =   \mathcal{R}_{x-y} \cap \mathcal{A}_x$, where $y\in \mathcal{N}$.  $\mathcal{D}_{x,x} =   \mathcal{R}_x \cap \mathcal{A}_x = \text{span}(\{\vert a_x \rangle \langle a_x \vert \otimes \vert b_x \rangle \langle b_x \vert \})$, where $ \vert a_x \rangle \in \{\vert 0 \rangle, \vert 1 \rangle\}$, and $ \vert b_x \rangle \in \{\vert + \rangle, \vert - \rangle\}$. $\mathcal{D}_{x-1,x} =   \mathcal{R}_{x-1} \cap \mathcal{A}_x = \mathbb{C} \mathbb{I} \otimes \mathbb{I}$. $\mathcal{D}_{x+1,x} =   \mathcal{R}_{x+1} \cap \mathcal{A}_x = \mathbb{C} \mathbb{I} \otimes \mathbb{I}$. Since  dim $\begin{rm}{span }\end{rm}(\prod_{y \in  \{-1,0,1\}} \mathcal{D}_{{x-y},x}) = 4$, while dim $ \mathcal{A}_x = 16$,   Theorem~\ref{thmStructure}  implies that  this QCA is again not a QLGA.

The change in neighborhood on quantization for this cellular automaton should lead us to question the original definition. Is there some redefinition of the cells that can maintain the same neighborhood for classical and quantum automata, and also yield a QLGA? We can in fact take the classical CA of Figure~\ref{qlgafig3} and change it slightly by allowing a shift of bits, as shown in Figure~\ref{qlgafig4}. That is, we have an extra step  that shifts the $a_{x+1}$ to $a_x$ before applying  the add operation. The result is the following CA:
\begin{align*}
\mu: a_{x} &\mapsto  a_{x+1}  \\
 b_{x}  &\mapsto  b_x +  a_{x} + a_{x+1}
\end{align*}
Then, after quantization,  we obtain a QLGA from it. We accomplish that  because the shift of bits corresponds to  a shift of tensor factors (advection). This QLGA has a neighborhood $\mathcal{N} = \{0,1\}$, i.e., is a radius-$1/2$ QLGA.
\begin{align*}
R: \vert a_{x} \rangle &\mapsto \vert a_{x+1} \rangle \\
\vert b_{x} \rangle &\mapsto \vert b_{x} + a_{x} + a_{x+1} \rangle 
\end{align*}

 \begin{figure}[htbp] 
\centering 
\includegraphics{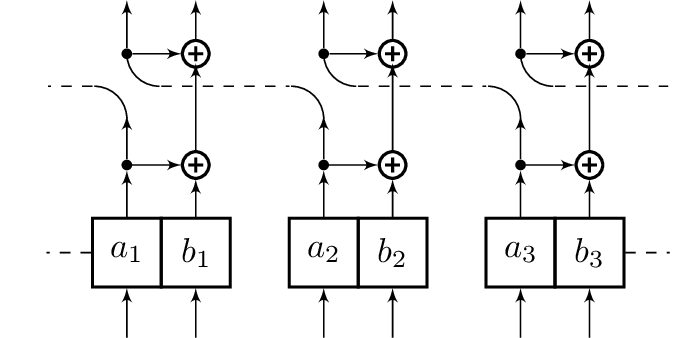}
\caption{One-dimensional radius-$1/2$ QLGA from permutation of bits in the CA of Figure~\ref{qlgafig3}}
\label{qlgafig4}
\end{figure}

\section{Conclusions}

We have given a local algebraic condition for a quantum cellular automaton to be a quantum lattice-gas. This result classifies QLGA as a subset of QCA, and resolves the question of which QCA are equivalent to a circuit with a ``brick structure''. This question has been open since Arrighi {\em et al.}~\cite{anw:odqca} produced an example showing that QLGA and QCA are distinct classes, contrary to the work of ~\cite{sw:rvqca}.  The question of classifying the remaining QCA that are not QLGA remains open. This open question motivates the reexamination of the definition of QCAs through local rules~\cite{a:watrous,dls:dpwfqca,ds:dpulqca,dm:u1dqca}. The techniques introduced in this paper allow us to obtain rigorous results for models defined on the Hilbert space of finite, unbounded configurations. The definition of this space necessarily requires the definition of products with countably infinitely many terms~\cite{vn:idp,ag:shrt}, just as is required for the constructions for WQCA defined through local rules~\cite{a:watrous,dls:dpwfqca,ds:dpulqca,dm:u1dqca}. One might hope that the techniques that are used in the present paper may also have utility in deciding which quantum automata defined through local rules are also causal, and therefore true QCA.

\section*{Acknowledgements}

The authors would like to acknowledge productive discussions with David Meyer,  Nolan Wallach and Ji\v{r}i Lebl. This project is supported by NSF CCI center, ``Quantum Information for Quantum Chemistry (QIQC)", award number CHE-1037992, and by NSF award PHY-0955518.

\appendix
\section{Some background and results  on representations of associative  algebras} \label{apprt}

We give a brief primer on  representations of associative algebras relevant to this paper. The material in this section, except for Corollary~\ref{corSS} and Theorem~\ref{vndnsthm}, is drawn from the texts by Goodman and Wallach~\cite{wall:sri, wall:ricg} and condensed for our purposes. Vector spaces in this section  are assumed to be  complex.  A representation of an associative algebra $\mathcal{A}$ is a pair $(\gamma,V)$ such that $\gamma: \mathcal{A} \longrightarrow \text{End}(V)$ is an associative algebra isomorphism. Since  a representation $(\gamma, V)$ of  $\mathcal{A}$  makes $V$ an $\mathcal{A}$-module under the action: $av = \gamma(a) v$, for $a \in \mathcal{A}$, and $v\in V$, we  can call  $V$ an $\mathcal{A}$-module instead (with $\gamma$ implied by the context).  If $U$ is a linear subspace of $V$ such that $\gamma(a) U \subset U$ for all $a \in \mathcal{A}$, then $U$ is \textit{invariant} under the representation. A representation $(\gamma,V)$ is \textit{irreducible} if the only invariant subspaces are $\{0\}$ and $V$. 

 Let  $(\gamma, V)$, $(\mu, W)$ be two  representations of an associative algebra  $\mathcal{A}$.  Let $\begin{rm}{Hom}\end{rm}(V,W)$ be the space of $\mathbb{C}$-linear maps from $V$ to $W$. Denote by  $\begin{rm}{Hom}_{\mathcal{A}}\end{rm}(V,W)$ the set of all $T \in \begin{rm}{Hom}\end{rm}(V,W)$ such that $T \gamma(a) = \mu(a)T$ for all $a \in \mathcal{A}$. Such a map is called an \textit{intertwining operator} between  the two representations. Two representations $(\gamma, V)$, $(\mu, W)$ are \textit{equivalent} if there exists an invertible intertwining operator between the two representations.

\begin{lemma}[Schur's Lemma, Lemma $4.1.4$, pg. 180,  Goodman and Wallach~\cite{wall:sri}]  \label{srlem}
Let $(\gamma,V)$ and $(\mu, W)$ be irreducible representations of an associative algebra $\mathcal{A}$. Assume that $V$ and $W$ have countable dimension over $\mathbb{C}$. Then 
\begin{equation*}
\begin{rm}{dim  }\:\end{rm}  \begin{rm}{Hom}_{\mathcal{A}}\end{rm}(V,W) = \left\{
\begin{array}{lr} 
1, & \text{if } (\gamma,V) \cong (\mu, W)\\
0, & \text{otherwise}
\end{array}
\right .
\end{equation*}
\end{lemma}

Let $(\gamma,V)$ be a finite-dimensional representation of $\mathcal{A}$. Let $\widehat{\mathcal{A}}$ be the set of all equivalence classes of finite-dimensional irreducible representations of $\mathcal{A}$. For each $\lambda \in \widehat{\mathcal{A}}$, fix a model  $(\pi^\lambda, V^\lambda)$.  
Let:
\begin{equation} \label{Udef}
U^\lambda = \begin{rm}{Hom}_{\mathcal{A}}\end{rm}(V^\lambda,V)
\end{equation}
For each $\lambda \in \widehat{\mathcal{A}}$, $U^\lambda$ is   called a \textit{multiplicity space}.  
 
Define the map:
\begin{align} \label{Slambda}
S_\lambda : U^\lambda \otimes V^\lambda &\longrightarrow V \\
 u^\lambda \otimes  v^\lambda &\mapsto  u^\lambda( v^\lambda) \nonumber
\end{align}
Then $S_\lambda$ is an intertwining operator with $ U^\lambda \otimes V^\lambda$ an $\mathcal{A}$-module under the action $a.(u \otimes v) = u \otimes (av)$ for $a \in \mathcal{A}$. 

For $\lambda \in \widehat{\mathcal{A}}$, define the $\lambda$-\textit{isotypic component}:
\begin{align*} 
V^{[\lambda]} :\sum_{W \subset V :  W\sim \lambda }W 
\end{align*}

A finite-dimensional $\mathcal{A}$-module  $V$ is called \textit{completely reducible} if for every invariant subspace $W \subset V$ there is a complementary invariant subspace $U$ such that $V=W\oplus U$. The following result (Proposition $4.1.11$, pg. 183, in~\cite{wall:sri}) provides equivalent characterizations of completely reducible representations.
\begin{prop} \label{compred} Let $(\gamma,V)$ be a finite-dimensional representation of $\mathcal{A}$. The following are equivalent:
\begin{enumerate}[label=(\roman{*})] 
\item \label{compred1} $V$ is completely reducible.
\item \label{compred2} $V = V_1 \oplus  \cdots \oplus V_s$ with each $V_i$ an irreducible $\mathcal{A}$-module.
\item \label{compred3} $V = W_1 +  \cdots + W_d$ with each $W_j$ an irreducible $\mathcal{A}$-module.
\end{enumerate}
\end{prop}

Let $(\gamma,V)$ be a finite-dimensional representation of $\mathcal{A}$, such that $V$ is a completely reducible $\mathcal{A}$-module. 
Then it is a result from the representation theory of associative algebras (Goodman and Wallach~\cite{wall:ricg}, Proposition $3.1.6$, pg. $120$), that the map $S_\lambda$ in~\eqref{Slambda} is an  $\mathcal{A}$-module  isomorphism of $U^{\lambda}\otimes V^\lambda$ with the submodule $V^{[\lambda]} \subset V$:
\begin{align} \label{Ulambdaisom} 
U^{\lambda}\otimes V^\lambda  \cong  V^{[\lambda]}
\end{align}
Further:
\begin{align*} 
V = \bigoplus_{\lambda \in \widehat{\mathcal{A}}} V^{[\lambda]}
\end{align*}
This is called the \textit{primary decomposition} of $V$. Thus  if we extend the maps $S_\lambda$ as follows:
\begin{align}  \label{Sdef}
S_\gamma = \bigoplus_{\lambda \in \widehat{\mathcal{A}}} S_\lambda :\bigoplus_{\lambda \in \widehat{\mathcal{A}}} U^{\lambda}\otimes V^\lambda  &\longrightarrow V \\
 \sum_{\lambda \in \widehat{\mathcal{A}}}  u^\lambda \otimes  v^\lambda &\mapsto  \sum_{\lambda \in \widehat{\mathcal{A}}}  u^\lambda( v^\lambda) \nonumber
\end{align}
then  $S_\gamma$ in~\eqref{Sdef} is an isomorphism of $\mathcal{A}$-modules (where  $V$ is an $\mathcal{A}$-module via $\gamma$). Then we have the following isomorphism of algebras.
\begin{align}
S_\gamma^{-1} \gamma(a)\ S_\gamma = \bigoplus_{\lambda \in \widehat{\mathcal{A}}}  \mathbb{I}_{U^\lambda} \otimes \pi^\lambda(a)
\end{align}
where $a \in \mathcal{A}$, $U^\lambda$ is as in~\eqref{Udef}, and $(\pi^\lambda , V^\lambda)$ is a model for the  class $\lambda$.

We state  the Double Commutant Theorem, in Goodman and Wallach~\cite{wall:sri} pg. $184$, specialized here for a finite-dimensional vector space.   Let $V$ be a finite-dimensional vector space. Define, for any subset $\mathcal{S} \subset \text{End}(V)$,  the \textit{commutant} of  $\mathcal{S}$:
\begin{equation} \label{commdef}
\text{Comm}(\mathcal{S}) = \{ y \in \text{End}(V): ys = sy \:\: \forall s \in \mathcal{S} \}
\end{equation}

\begin{thm}[Double Commutant, Theorem $4.1.13$~\cite{wall:sri}] \label{dblctthmcr}
 Suppose  $\mathcal{A} \subset \begin{rm}{End}(V)\end{rm}$   is a  subalgebra  containing the identity operator, such that   $V$ is  a completely reducible $\mathcal{A}$-module.   Set $\mathcal{B} =  \begin{rm}{Comm}\end{rm}(\mathcal{A})$. Then $\mathcal{A} =  \begin{rm}{Comm}\end{rm}(\mathcal{B})$. 
\end{thm}

Denote by $\begin{rm}{Spec}(\gamma)\end{rm}$ the set of irreducible representations that occur in the primary decomposition of $V$. For each 
$\lambda \in \begin{rm}{Spec}(\gamma)\end{rm}$, the multiplicity space $U^\lambda = \begin{rm}{Hom}_{\mathcal{A}}\end{rm}(V_\lambda,V)$ in~\eqref{Sdef} is a $\begin{rm}{Comm}\end{rm}(\mathcal{A})$-module under the action:
 \begin{equation} \label{bmodact}
b.u = bu
\end{equation}
where $b \in \begin{rm}{Comm}\end{rm}(\mathcal{A})$, and $u \in  U^\lambda$, and $bu: V^\lambda \underrightarrow{u} V  \underrightarrow{b} V$ is the composition (left multiplication by $b$) map. 

We adapt  some results from Goodman and Wallach~\cite{wall:sri}, Section $4.2.1$ (General Duality Theorem). Except for  some substitutions, these results are taken exactly as in the said reference, including the proofs. We obtain from  the originals, derived in the setting of group representations,   the versions relevant to   subalgebras of  $\begin{rm}{End}(V)\end{rm}$.   Thus we get  the associative algebra counterparts of the   \textit{Duality Theorem}, Theorem $4.2.1$~\cite{wall:sri}, and  Corollary $4.2.4$ in~\cite{wall:sri}.

\begin{thm}[Duality, Theorem $4.2.1$~\cite{wall:sri}] \label{dualitythm}
 Suppose  $\mathcal{A} \subset \begin{rm}{End}(V)\end{rm}$   is an  subalgebra  containing the identity operator. Let $V$ be a completely reducible $\mathcal{A}$-module. Each multiplicity space $U^{\lambda}$ in~\eqref{Sdef} is an irreducible $\begin{rm}{Comm}\end{rm}(\mathcal{A})$-module. Furthermore, if $\lambda, \mu \in \widehat{\mathcal{A}}$, and $U^{\lambda} \cong U^{\mu}$ as $\begin{rm}{Comm}\end{rm}(\mathcal{A})$-module, then $\lambda=\mu$.
\end{thm}
\begin{proof}
We first prove that the action of $\begin{rm}{Comm}\end{rm}(\mathcal{A})$ on $U^{\lambda}$ is irreducible.  Let $T \in U^{\lambda}$ be nonzero. Given another nonzero element $S \in U^{\lambda}$ we need to find $r \in \begin{rm}{Comm}\end{rm}(\mathcal{A})$ such that $rT = S$. Let $X = T V^{\lambda}$ and $Y = S V^{\lambda}$. Then by Schur's Lemma (Lemma~\ref{srlem}), $X$ and $Y$ are isomorphic $\mathcal{A}$-modules of class $\lambda$.  Thus there exists $u \in \begin{rm}{Comm}\end{rm}(\mathcal{A})$ such that $uT : V^{\lambda} \longrightarrow SV^{\lambda}$ is an $\mathcal{A}$-module isomorphism. Schur's Lemma implies that there exists $c \in \mathbb{C}$ such that $cuT = S$, so we may take $r = cu$. 

We now show that if $\lambda \neq \mu$ then $U^{\lambda}$ and $U^{\mu}$ are inequivalent modules for $\begin{rm}{Comm}\end{rm}(\mathcal{A})$.  Suppose 
\begin{equation*}
\phi: U^{\lambda} \longrightarrow U^{\mu}
\end{equation*}
is an intertwining operator for the action of $\begin{rm}{Comm}\end{rm}(\mathcal{A})$. Let $T\in U^{\lambda}$ be nonzero and set $S = \phi(T)$. We want to show that $S = 0$. Set $U = T V^{\lambda} + S V^{\mu}$. Then since we are assuming $\lambda \neq \mu$, the sum is direct. Let $p: U \longrightarrow S V^{\mu}$ be the corresponding projection.  Then $p \in \begin{rm}{Comm}\end{rm}(\mathcal{A})$. Since $pT = 0$, we have: 
\begin{equation*}
0 = \phi(pT) = p\phi(T) = pS = S\text{,}
\end{equation*}
which proves that $\phi = 0$. 
\end{proof}

\begin{cor}[Corollary $4.2.4$~\cite{wall:sri}] \label{dualcor}
Let $\mathcal{B} =  \begin{rm}{Comm}\end{rm}(\mathcal{A})$. Then $V$ is a completely reducible $\mathcal{B}$-module. Furthermore, the following hold:
\begin{enumerate}[label=(\roman{*})] 
\item \label{dualcor1} Suppose for every $\lambda \in \begin{rm}{Spec}(\gamma)\end{rm}$ there is given an operator $T_\lambda \in \begin{rm}{End}\end{rm}(V^\lambda)$. Then there exists $T \in \mathcal{A}$ that acts by $\mathbb{I} \otimes T_\lambda$ on the $\lambda$-summand in the decomposition~\eqref{Sdef}.
\item \label{dualcor2}  Let $T \in \mathcal{A} \cap \mathcal{B}$ (the center of $\mathcal{A}$). Then $T$ is diagonalized by~\eqref{Sdef} and acts by scalar ${\hat T}(\lambda) \in \mathbb{C}$ on $U^\lambda \otimes V^\lambda$. Conversely, given any complex valued function $f$ on $\begin{rm}{Spec}(\gamma)\end{rm}$, there exists $T \in \mathcal{A} \cap \mathcal{B}$ such that ${\hat T}(\lambda) = f(\lambda)$.
\end{enumerate}  
\end{cor}
\begin{proof}
Since $V$ is the direct sum of $\mathcal{B}$-invariant irreducible subspaces by Theorem~\ref{dualitythm}, it is a completely reducible $\mathcal{B}$-module by Proposition~\ref{compred}. We now prove the other assertions.
\begin{enumerate}[label=(\roman{*})] 
\item Let $T \in \begin{rm}{End}(V)\end{rm}$ be the operator that acts by $\mathbb{I} \otimes T_\lambda$ on the $\lambda$ summand. Then $T \in \begin{rm}{Comm}\end{rm}(\mathcal{B})$, and hence $T \in  \mathcal{A}$ by the Double Commutant Theorem (Theorem~\ref{dblctthmcr}).
\item Each summand in~\eqref{Sdef} is invariant under $T$, and the action of $T$ on the $\lambda$ summand is by an operator of the form $R_\lambda \otimes \mathbb{I} = \mathbb{I} \otimes S_\lambda$ with $R_\lambda \in \begin{rm}{End}(U^\lambda)\end{rm}$, and $S_\lambda \in \begin{rm}{End}(V^\lambda)\end{rm}$. Such an operator must be a scalar multiple of the identity operator. The converse follows from~\ref{dualcor1}.
\end{enumerate}  

\end{proof}

An algebra $\mathcal{B}$ is  \textit{simple} if  the only two sided ideals in  $\mathcal{B}$ are $0$ and $\mathcal{B}$.  An algebra $\mathcal{A}$ is a \textit{semisimple} algebra if  it is  a finite direct sum of  \textit{simple} algebras. This is equivalent, by Wedderburn's Theorem (Goodman and Wallach~\cite{wall:ricg}, Theorem $3.2.1$, pg. $128$),  to the statement  that $\mathcal{A}$ is isomorphic to a finite sum of matrix algebras.  

For an algebra $\mathcal{A} \subset \begin{rm}{End}\end{rm}(V)$, such that $V$ is a completely reducible $\mathcal{A}$-module,  let  $\begin{rm}{Spec}\end{rm}(\mathcal{A})$ be the set of irreducible representations that occur in the primary decomposition of $V$. 

\begin{cor} \label{corSS}
 Let $\mathcal{A} \subset \begin{rm}{End}\end{rm}(V)$ be a  subalgebra containing the identity operator, such that $V$ is a completely reducible $\mathcal{A}$-module. Let   $\mathcal{B} = \begin{rm}{Comm}\end{rm}(\mathcal{A})$.  Then there is an $\mathcal{A}$-module isomorphism:
\begin{align}  \label{Aisom}
S_V: \bigoplus_{\lambda \in \widehat{\mathcal{A}}} U^{\lambda}\otimes V^\lambda  &\longrightarrow V \\
 \sum_{\lambda \in \widehat{\mathcal{A}}}  u^\lambda \otimes  v^\lambda &\mapsto  \sum_{\lambda \in \widehat{\mathcal{A}}}  u^\lambda( v^\lambda) \nonumber
\end{align}
where $V^\lambda$ is an irreducible module for the class $\lambda \in \widehat{\mathcal{A}}$, and $U^\lambda = \begin{rm}{Hom}_{\mathcal{A}}\end{rm}(V^\lambda,V)$. Under this isomorphism:
\begin{align*} 
\mathcal{A} \cong \bigoplus_{\lambda \in \widehat{\mathcal{A}}}  \mathbb{I}_{U^\lambda} \otimes \begin{rm}{End}\end{rm}(V^{\lambda})
\end{align*}
and:
\begin{equation*}
\mathcal{B} \cong \bigoplus_{\lambda \in \widehat{\mathcal{A}}}   \begin{rm}{End}\end{rm}(U^\lambda)  \otimes \mathbb{I}_{V^\lambda}
\end{equation*}
where the action of $\mathcal{B}$ on each $U^\lambda$ is given by~\eqref{bmodact}.

In particular,  $\mathcal{A}$ is semisimple.
\end{cor}
\begin{proof} $V$ is a completely reducible $\mathcal{A}$-module. This implies there is an $\mathcal{A}$-module isomorphism given by~\eqref{Sdef}:
\begin{align*}
S_V: \bigoplus_{\lambda \in \widehat{\mathcal{A}}} U^{\lambda}\otimes V^\lambda  &\longrightarrow V \\
 \sum_{\lambda \in \widehat{\mathcal{A}}}  u^\lambda \otimes  v^\lambda &\mapsto  \sum_{\lambda \in \widehat{\mathcal{A}}}  u^\lambda( v^\lambda) \nonumber
\end{align*}
where $V^\lambda$ is an irreducible module for the class $\lambda \in \widehat{\mathcal{A}}$, and $U^\lambda = \begin{rm}{Hom}_{\mathcal{A}}\end{rm}(V^\lambda,V)$.  Then by definition of $\mathcal{A}$ action (as a subalgerba of  $\begin{rm}{End}\end{rm}(V)$) on $V$,  under this isomorphism, Corollary~\ref{dualcor} implies:
\begin{align*} 
\mathcal{A} \cong \bigoplus_{\begin{rm}{Spec}\end{rm}(\mathcal{A})}  \mathbb{I}_{U^\lambda} \otimes \begin{rm}{End}\end{rm}(V^{\lambda})
\end{align*}
Since $V$ is finite-dimensional, the above implies that $\begin{rm}{Spec}\end{rm}(\mathcal{A})$ is a finite set. This  in turn implies that $\mathcal{A}$ is semisimple as it is isomorphic to a finite direct sum of finite dimensional matrix algebras. Then a standard theorem on semisimple algebras, Proposition $3.3.1$ in~\cite{wall:ricg}, implies that each irreducible representations of $\mathcal{A}$ is equivalent to some element of $\begin{rm}{Spec}\end{rm}(\mathcal{A})$. Hence $\widehat{\mathcal{A}} = \begin{rm}{Spec}\end{rm}(\mathcal{A})$.  Therefore:
\begin{align*} 
\mathcal{A} \cong \bigoplus_{\lambda \in \widehat{\mathcal{A}}}  \mathbb{I}_{U^\lambda} \otimes \begin{rm}{End}\end{rm}(V^{\lambda})
\end{align*}
By Duality Theorem (Theorem~\ref{dualitythm}), under $S_V$:
\begin{equation*}
\mathcal{B} \cong \bigoplus_{\lambda \in \widehat{\mathcal{A}}}   \begin{rm}{End}\end{rm}(U^\lambda)  \otimes \mathbb{I}_{V^\lambda}
\end{equation*}
where the action of $\mathcal{B}$ on each $U^\lambda$ is given by~\eqref{bmodact}.
\end{proof}

We state  the semisimple algebra version of the Double Commutant Theorem, in Goodman and Wallach~\cite{wall:ricg} pg. $137$.  
\begin{thm}[Double Commutant Theorem for Semisimple Algebras]\label{dblctthm}
 
Suppose  $\mathcal{A} \subset \begin{rm}{End}(V)\end{rm}$   is a semisimple subalgebra  containing the identity operator. 
Then the algebra  $\mathcal{B} =  \begin{rm}{Comm}\end{rm}(\mathcal{A})$  is semisimple and $\mathcal{A} =  \begin{rm}{Comm}\end{rm}(\mathcal{B})$. Furthermore, there exists an $\mathcal{A}$-module isomorphism:
\begin{align*}  
S_\mathcal{A} : \bigoplus^r_{j=1}U_j  \otimes V_j  & \longrightarrow  V  \\
 \sum^r_{j=1}  u_j \otimes  v_j &\mapsto  \sum^r_{j=1}  u_j( v_j) \nonumber
\end{align*}
where $V_j$ is an irreducible $\mathcal{A}$-module, and $U_j =  \begin{rm}{Hom}_{\mathcal{A}}\end{rm}(V_j,V)$.  Under this isomorphism:
\begin{equation*}
\mathcal{A} \cong  \bigoplus^r_{j=1} \mathbb{I}_{U_j}  \otimes \begin{rm}{End}\end{rm}(V_j)
\end{equation*}
and:
\begin{equation*}
\mathcal{B} \cong  \bigoplus^r_{j=1}   \begin{rm}{End}\end{rm}(U_j)  \otimes \mathbb{I}_{V_j}
\end{equation*}
\end{thm}

Next we consider a general Hilbert space,  and state the remarkable von Neumann Density Theorem.  The  proof of the von Neumann Density Theorem  can be found in~\cite{kr:ftoa,wall:rrg}. Consider a Hilbert space $\mathcal{H}$ and the space, $B(\mathcal{H})$,  of bounded linear operators of $\mathcal{H}$. 
Recall the following definition.  For any subset $\mathcal{S} \subset B(\mathcal{H})$,  the \textit{commutant} of  $\mathcal{S}$ is defined as:
\begin{equation*} %\label{commdefcd}
\text{Comm}(\mathcal{S}) = \{ y \in B(\mathcal{H}): ys = sy \:\: \forall s \in \mathcal{S} \}
\end{equation*}
\begin{thm}[von Neumann Density Theorem]  \label{vndnsthm}
Let  $\mathcal{A} \subset B(\mathcal{H})$ be a self-adjoint algebra of operators containing the identity operator, then the weak  and strong closures of $\mathcal{A}$ in $B(\mathcal{H})$ is  $\begin{rm}{Comm}\end{rm}(\begin{rm}{Comm}\end{rm}(\mathcal{A}))$.
\end{thm}
 
 \section{Some results on representations of self-adjoint algebras} \label{saalg}
 We state a few  results on representations of self-adjoint algebras. Let $V$ be a finite-dimensional Hilbert space, with an inner product $\langle \cdot |  \cdot \rangle$, and   $\mathcal{A} \subset \begin{rm}{End}\end{rm}(V)$ a subalgebra. Let the adjoint of an  algebra $\mathcal{A} \subset \begin{rm}{End}\end{rm}(V)$ be $\mathcal{A}^\dag = \{A^\dag : A \in \mathcal{A}\}$. An algebra $\mathcal{A}$  is self-adjoint if $\mathcal{A} = \mathcal{A}^\dag$. From the results in this section, Lemma~\ref{lemmaSA}, and Proposition~\ref{propSA} appear in~\cite{cp:lg} (pg. 145).
\begin{lemma} \label{lemmaSA}
 Let $\mathcal{A} \subset \begin{rm}{End}\end{rm}(V)$ be a self-adjoint subalgebra. If $W \in V$ is an $\mathcal{A}$-invariant subspace, then $W^\perp = \{ {v} \in V : \langle v | w \rangle = 0 \: \forall \:  {w} \in W\}$ is $\mathcal{A}$-invariant. 
\end{lemma}

\begin{proof}
Let $ {w} \in W$, $ {v} \in  W^\perp$, $A \in \mathcal{A}$. Then $A^\dag \in \mathcal{A}$, which implies $A^\dag  {w} \in W$ $\implies$   $\langle A  {v} | w \rangle  = \langle  v | A^\dag w \rangle = 0$ $\implies$  $A  {v} \in W^\perp$.
\end{proof}

\begin{prop} \label{propSA}
 Let $\mathcal{A} \subset \begin{rm}{End}\end{rm}(V)$ be a self-adjoint subalgebra. Then $V$ is an orthogonal direct sum of irreducible  $\mathcal{A}$-modules. In particular, $V$ is a completely reducible $\mathcal{A}$-module.
\end{prop}

\begin{proof}
Let   $W \subset V$ be an  $\mathcal{A}$-invariant subspace of minimal dimension.  Then it is by definition irreducible. Since $ \mathcal{A} =  \mathcal{A}^\dag$, by Lemma~\ref{lemmaSA},  $V = W \oplus W^\perp$ is an orthogonal direct sum  of $\mathcal{A}$-modules. The conclusion follows by  induction on dimension.
\end{proof}

\begin{cor} \label{sass}
 Let $\mathcal{A} \subset \begin{rm}{End}\end{rm}(V)$ be a self-adjoint subalgebra  containing the identity operator. Then $\mathcal{A}$ is semisimple.
\end{cor}
\begin{proof}
By Proposition~\ref{propSA}, $V$ is a completely reducible $\mathcal{A}$-module. Then Corollary~\ref{corSS} implies that $\mathcal{A}$ is semisimple.
\end{proof}

\begin{cor}  \label{corSSB}
Let $\mathcal{A} \subset \begin{rm}{End}\end{rm}(V)$ be a self-adjoint subalgebra.  Then there exist for every $\lambda  \in  \widehat{\mathcal{A}}$,  a set of intertwining  operators $\{u^\lambda_j\} \subset  U^\lambda$, such that the $\lambda$-isotypic component can be written as an orthogonal direct sum: $V^{[\lambda]} =  \bigoplus_{j} u^\lambda_j(V^\lambda)$. 
\end{cor}

 \begin{proof}
By Proposition~\ref{propSA}, $V$ is a completely reducible $\mathcal{A}$-module. Thus by~\eqref{Ulambdaisom}, $S_\lambda$ in~\eqref{Slambda} is an  $\mathcal{A}$-module  isomorphism:
\begin{align*} 
S_\lambda: U^{\lambda} \otimes V^\lambda  \longrightarrow  V^{[\lambda]}
\end{align*}
Again by  Proposition~\ref{propSA}, the $\lambda$-isotypic component,
\begin{align*} 
V^{[\lambda]} = \bigoplus_{W_j \subset V :  W_j\sim \lambda }W_j 
\end{align*}
is an orthogonal direct sum.  Choose $u_j$ such that $u_j(V^\lambda) = W_j$.
\end{proof}

\section{Structural Reversibility }
  
We prove the  Structural Reversibility  (due to Arrighi, Nesme, and Werner~\cite{anw:odqca}), in our context,  to make the present paper self-contained. 
\begin{thm}[Structural Reversibility ] \label{strucrev}
Let $M : \mathcal{H}_{\mathcal{C}} \longrightarrow \mathcal{H}_{\mathcal{C}}$ be a unitary  operator  and $\mathcal{N}$ a neighborhood. Let $\mathcal{V} = \{-k | k\in \mathcal{N}\}$. Then the following are equivalent:
\begin{enumerate}[label=(\roman{*})] 
\item \label{strv1} $M$ is causal relative to the  neighborhood $\mathcal{N}$. 
\item \label{strv2}  For every  operator $A_x$ local upon cell $x$, $M^\dag A_x M$ is local upon  $\mathcal{N}_x$.
\item \label{strv3} $M^\dag$ is causal relative to the  neighborhood $\mathcal{V}$. 
\item \label{strv4}  For every  operator $A_x$ local upon cell $x$, $M A_x M^\dag$ is local upon  ${\mathcal{V}}_x$.
\end{enumerate}   
\end{thm}
 
Before proving the theorem, we prove the following lemma:

\begin{lemma} \label{lemmatr}
Denote the set of density operators on  $\mathcal{H}_{\mathcal{C}}$ by $\Omega$.  For  $A \in B(\mathcal{H}_{\mathcal{C}})$, define the map:
\begin{align*}
 f_A : \Omega \longrightarrow \mathbb{C} \\
 f_A(\rho) = \begin{rm}{tr}\end{rm}(A \rho)
 \end{align*}
 where $\begin{rm}{tr}\end{rm}$ is the trace operator. Let $D \in \mathbb{Z}^n$ be a finite subset. Then:
 \begin{enumerate}[label=(\roman{*})] 

\item \label{lemmatr1} $A$ is local upon $D$ if and only if for every pair $\rho, \rho' \in \Omega$ satisfying  $\rho_{|_D} = {\rho'}_{|_D}$, $f_A(\rho) = f_A(\rho')$. 
\item \label{lemmatr2} $\rho, \rho' \in \Omega$ satisfy $\rho_{|_D} = {\rho'}_{|_D}$, if and only if for every $A$ local upon $D$, $f_A(\rho) = f_A(\rho')$.

\end{enumerate}
\end{lemma}
 \begin{proof}
For elements of $B(\mathcal{H}_{\mathcal{C}})$ that are local on some finite  subset $K \subset \mathbb{Z}^n$,   by abuse of notation, we call both $B \in \mathcal{Z}$ ($\mathcal{Z}$ as in~\eqref{Zdef}) and $\iota^{-1}_K(B)$ ($\iota_K$ as in~\eqref{algembd}) by $B$. 

 \ref{lemmatr1} Suppose $A$ is local upon $D$ and $\rho, \rho' \in \Omega$ are such that $\rho_{|_D} = {\rho'}_{|_D}$. Then  $f_A(\rho) = \begin{rm}{tr}\end{rm}(A \rho) = \begin{rm}{tr}\end{rm}(A \rho_{|_D}) = \begin{rm}{tr}\end{rm}(A {\rho'}_{|_D})  = \text{tr}(A \rho') = f_A(\rho')$.

For the converse,  suppose $A$ is not local upon $D$.  We can choose an orthonormal basis $\{\vert v_{j} \rangle\}$ for $ \bigotimes_{k \in  D} W$. Then we can write (special case of ~\eqref{opbrk}):
\begin{equation*}
A = \sum_{(l,m)} \vert v_l \rangle \langle v_m \vert \otimes A_{l,m}
\end{equation*}
where  $\{(l,m)\} \subset \{1,\ldots,d_W\} \times \{1,\ldots,d_W\}$, and $A_{l,m} \in B(\mathcal{H}_{\mathcal{C}_{\overline{D}}})$ are non-zero.  
Since $A$ is not local upon $D$, not all the $A_{l,m}$ are multiples of  the identity, $ \mathbb{I}_{\overline{D}}$, on $\mathcal{H}_{\mathcal{C}_{\overline{D}}}$. There are two cases to consider:

\begin{enumerate} 
 \item There  is an $A_{l,l}$ for some $l$ such that $A_{l,l}$ is not a multiple of $\mathbb{I}_{\overline{D}}$. Then there exist two unit vectors $\vert x \rangle, \vert y \rangle \in \mathcal{H}_{\mathcal{C}_{\overline{D}}}$, such that $  \langle x \vert A_{l,l} \vert x \rangle \neq \langle y \vert A_{l,l} \vert y \rangle$. Let $\rho = \vert v_l \rangle \langle v_l \vert \otimes \vert x \rangle \langle x \vert, \rho' = \vert v_l \rangle \langle v_l \vert \otimes \vert y \rangle \langle y \vert$. Then $\rho_{|_D} = {\rho'}_{|_D}$ whereas $f_A(\rho) \neq f_A(\rho')$.

\item Every  $A_{l,l}$   is  a multiple of  $\mathbb{I}_{\overline{D}}$.  Then choose some $A_{l,m}$  which is not. Then there exist two unit vectors $\vert x \rangle, \vert y \rangle \in \mathcal{H}_{\mathcal{C}_{\overline{D}}}$, such that $  \langle x \vert A_{l,m} \vert x \rangle \neq \langle y \vert A_{l,m} \vert y \rangle$. Let $\rho = (\vert v_l \rangle \langle v_l \vert + \vert v_l \rangle \langle v_m \vert) \otimes \vert x \rangle \langle x \vert, \rho' = (\vert v_l \rangle \langle v_l \vert + \vert v_l \rangle \langle v_m \vert) \otimes \vert y \rangle \langle y \vert$. Then $\rho_{|_D} = {\rho'}_{|_D}$ whereas $f_A(\rho) \neq f_A(\rho')$.
 \end{enumerate}

 \ref{lemmatr2} Suppose $A$ is local upon $D$ and $\rho, \rho' \in \Omega$ are such that $\rho_{|_D} = {\rho'}_{|_D}$. Then   $f_A(\rho) = \begin{rm}{tr}\end{rm}(A \rho) = \begin{rm}{tr}\end{rm}(A \rho_{|_D}) = \begin{rm}{tr}\end{rm}(A {\rho'}_{|_D})  = \text{tr}(A \rho') = f_A(\rho')$.
 
 For the converse, let $\rho, \rho' \in \Omega$,  such that for all $A$ local upon $D$, $f_A(\rho) = f_A({\rho'})$. Consider  $A  = \vert v_l \rangle \langle v_m \vert \otimes \mathbb{I}_{\overline{D}}$, a basis element of operators local upon $D$. Then  $f_A(\rho) =  \langle v_m \vert \rho_{|_D} \vert v_l \rangle$,  $f_A(\rho') =  \langle v_m \vert {\rho'}_{|_D} \vert v_l \rangle$. This implies by   $f_A(\rho) = f_A({\rho'})$,  that $ \langle v_m \vert \rho_{|_D} \vert v_l \rangle =  \langle v_m \vert {\rho'}_{|_D} \vert v_l \rangle$. Since this is true for  all elements of $\{\vert v_l \rangle \langle v_m \vert: 1\leq l,m  \leq d_W\}$, i.e., for a set of basis elements of operators local upon $D$, then by linearity, $\rho_{|_D} =  {\rho'}_{|_D}$.

\end{proof}

\begin{proof}[Proof of Theorem~\ref{strucrev}]
{\ \\}
Note:
\begin{equation} \label{treq}
\text{tr}(M\rho M^\dag A) =  \text{tr}(\rho M^\dag A M) :\:\ \forall \rho \in \Omega
\end{equation}

$~\ref{strv1} \implies ~\ref{strv2}$ Let $M$ be causal. Suppose $\rho$ and $\rho'$ is a pair such that $\rho|_{\mathcal{N}_x} = {\rho'}|_{\mathcal{N}_x}$, and  $A$ is local upon $x$. Causality implies $\text{tr}(M\rho M^\dag A) =  \text{tr}(M {\rho'} M^\dag A)$. By~\eqref{treq}, $\text{tr}(\rho M^\dag A M) =  \text{tr}({\rho'} M^\dag A M)$, so by  Lemma~\ref{lemmatr}~\ref{lemmatr1},  $M^\dag A M$ is local upon $\mathcal{N}_x$.

 $~\ref{strv2} \implies ~\ref{strv1}$ Suppose   $A$ is local upon $x$, $\rho$ and $\rho'$ be a pair of density operators such that $\rho|_{\mathcal{N}_x} = {\rho'}|_{\mathcal{N}_x}$, and  $M^\dag A M$ is local upon $\mathcal{N}_x$. Then $\text{tr}(\rho M^\dag A M) =  \text{tr}( {\rho'} M^\dag A M)$. By~\eqref{treq}, $\text{tr}(M\rho M^\dag A) =  \text{tr}(M{\rho'} M^\dag A)$, so by  Lemma~\ref{lemmatr}~\ref{lemmatr2},  $M \rho M^\dag|_x = M {\rho'} M^\dag|_x$.

$~\ref{strv2} \implies ~\ref{strv4}$ Suppose   $A_x$ is local upon $x$.  Let $A_y$ be an operator local upon $y \notin {\mathcal{V}}_x$. Then $M^\dag A_y M$ is local upon $\mathcal{N}_y$. But $x \notin \mathcal{N}_y$, so $M^\dag A_y M$ commutes with $A_x$. But then $A_y = M M^\dag A_y M M^\dag$ commutes with $M A_x M^\dag$, by the definition of conjugation. This applies to every $A_y \in \mathcal{A}_y$, hence $M A_x M^\dag$ commutes with $\mathcal{A}_y$. We can express $M A_x M^\dag$ as a finite sum (as in~\eqref{opbrk}):

\begin{equation} \label{Mxdec}
M A_x M^\dag = \sum_{l}  e_l  \otimes E_l
\end{equation}
where $\{e_l\} \subset \text{End}(\bigotimes_{k \in \mathcal{V}_x} W)$ ia a linearly independent set, and  $E_l \in B(\mathcal{H}_{\mathcal{C}_{\overline{\mathcal{V}_x}}})$ (the space of bounded linear operators on the co-$\mathcal{V}_x$ space, $\mathcal{H}_{\mathcal{C}_{\overline{\mathcal{V}_x}}}$,  defined in Definition~\ref{coDdef}, with $D = \mathcal{V}_x$).

For any $B \in \text{End}(\bigotimes_{k \in \mathcal{V}_x} W)$, $B \otimes \mathbb{I}_{\overline{\mathcal{V}_x}}$ (where $\mathbb{I}_{\overline{\mathcal{V}_x}}$ is the identity on $\mathcal{H}_{\mathcal{C}_{\overline{\mathcal{V}_x}}}$) also commutes with $A_y$, This implies $(B \otimes \mathbb{I}_{\overline{\mathcal{V}_x}})  (M A_x M^\dag)$  commutes with $A_y$. Using the standard  commutator notation $[.,.]$ ([X,Y]=XY-YX), we write:
\begin{equation*}
[A_y,(B \otimes \mathbb{I}_{\overline{\mathcal{V}_x}})  (M A_x M^\dag)] = \sum_{l} (B e_l)  \otimes [A_y,E_l] = 0
\end{equation*}
We can  trace out the finite dimensional part. Let the operator thus obtained on $\mathcal{H}_{\mathcal{C}_{\overline{\mathcal{V}_x}}}$ be $\zeta(B)$:
\begin{align*}
\zeta(B) =  \sum_{l}   \text{tr }(B e_l)   [A_y,E_l] =   \sum_{l}  \text{tr } (B e_l) [A_y,E_l] = 0
\end{align*}
Since $(X,Y) = \text{tr}(XY)$ is a nondegenerate  bilinear form on $ \text{End}(\bigotimes_{k \in \mathcal{V}_x} W)$, and $\{e_l\}$ is a  linearly independent  set, we can choose a  set $\{B_l\} \subset  \text{End}(\bigotimes_{k \in \mathcal{V}_x} W)$ (indexed by the same set $\{l\}$ that indexes $\{e_l\}$) such that $ \text{tr }(B_k e_l)  = \delta_{k,l}$.   Then:
\begin{equation} \label{commA}
\zeta(B_k) =   [A_y,E_k] = 0
\end{equation}

The subalgebra of $B(\mathcal{H}_{\mathcal{C}_{\overline{\mathcal{V}_x}}})$ generated by $\mathcal{A}_y$, $y \notin \mathcal{V}_x$ acts irreducibly on the space $\mathcal{H}_{\mathcal{C}_{\overline{\mathcal{V}_x}}}$ (by the argument in the proof of Theorem~\ref{Zdense}, i.e., it acts transitively on a dense subspace, $\text{span}({\mathcal{C}}_{\overline{\mathcal{V}_x}})$, of $\mathcal{H}_{\mathcal{C}_{\overline{\mathcal{V}_x}}}$).  By~\eqref{commA}, $E_k$ commutes with this subalgebra. Schur's Lemma (Lemma~\ref{srlem})   then implies that  $E_k = \eta_k \mathbb{I}_{\overline{\mathcal{V}_x}}$ for some $\eta_k \in \mathbb{C}$. This is true for every $k \in \{l\}$.  By~\eqref{Mxdec}:
\begin{equation*} 
M A_x M^\dag = (\sum_{l}  \eta_l e_l ) \otimes \mathbb{I}_{\overline{\mathcal{V}_x}}  = a_{\mathcal{V}_x} \otimes \mathbb{I}_{\overline{\mathcal{V}_x}}
\end{equation*}
for some $a_{\mathcal{V}_x} \in \text{End}(\bigotimes_{k \in \mathcal{V}_x} W)$.  Hence $M A_x M^\dag$ is local upon ${\mathcal{V}}_x$.

$~\ref{strv4} \implies ~\ref{strv2}$ Same argument as for the  proof of  $~\ref{strv2} \implies~\ref{strv4}$, except  interchanging the roles of  $\mathcal{N}_x$  and $\mathcal{V}_x$, $M$ and $M^\dag$.

$~\ref{strv3} \iff ~\ref{strv4}$ Same arguments as for the   proof of  $~\ref{strv1} \iff ~\ref{strv2}$, except  replacing  $\mathcal{N}$ with $\mathcal{V}$, $\mathcal{N}_x$ with $\mathcal{V}_x$, and $M$ with $M^\dag$.
\end{proof}

\section{Proof of Lemma~\ref{Fdefm} requiring quiescent symbol $\vert  \hat{q}_0 \rangle$ to be  an invariant of $F$} \label{ApplF}

\begin{lemmanonum}[Lemma~\ref{Fdefm}]
$\hat F$  is defined on  $\hat{\mathcal{C}}$ as a map $\hat F: \hat{\mathcal{C}}  \longrightarrow {\mathcal{H}}_{\hat{\mathcal{C}}}$, and   can be extended to  ${\mathcal{H}}_{\hat{\mathcal{C}}}$ as unitary  and  translation-invariant operator, if and only if   $F$ has  $\vert \hat{q}_0 \rangle$ as an invariant (an eigenvector with eigenvalue one): $F \vert \hat{q}_0 \rangle =  \vert \hat{q}_0 \rangle$.
 \end{lemmanonum}
 \begin{proof}[Proof of Lemma~\ref{Fdefm}]

Let us assume that $\hat F$  is defined on $\hat{\mathcal{C}}$ as a map $\hat F: \hat{\mathcal{C}}  \longrightarrow {\mathcal{H}}_{\hat{\mathcal{C}}}$, and   can be extended to  ${\mathcal{H}}_{\hat{\mathcal{C}}}$ as   a unitary  and  translation-invariant  operator.  Then Lemma~\ref{Rinvariant} implies:
\begin{equation*}
\hat{F}  \big( \bigotimes_{x \in \mathbb{Z}^n} \vert \hat{q}_0 \rangle \big) = \bigotimes_{x \in \mathbb{Z}^n} F \vert  \hat{q}_0 \rangle =  e^{i\Theta_0} \bigotimes_{x \in \mathbb{Z}^n} \vert \hat{q}_0 \rangle
\end{equation*}
for some $\Theta_0 \in \mathbb{R}$. This implies,  $F \vert \hat{q}_0 \rangle =  \vert \hat{q}_0 \rangle$.

For the   converse, assume $F \vert \hat{q}_0 \rangle =  \vert \hat{q}_0 \rangle$. Then for every element $\vert \hat{c} \rangle = \bigotimes_{x \in \mathbb{Z}^n}  \vert {\hat{c}}_x \rangle \in \hat{\mathcal{C}}$, under the formal definition of $\hat F$ in~\eqref{Fdefn}, $\hat{F}  \vert \hat{c} \rangle = \bigotimes_{x \in \mathbb{Z}^n}  F\vert {\hat{c}}_x \rangle \in {\mathcal{H}}_{\hat{\mathcal{C}}}$. Furthermore, as a map $\hat F: \hat{\mathcal{C}}  \longrightarrow {\mathcal{H}}_{\hat{\mathcal{C}}}$,  $\hat F$ preserves the inner product (since $F$ is a unitary operator on $\hat W = \bigotimes_{z \in \mathcal{N}} V_z$), and is  translation-invariant, by the same definition.  Then $\hat{F}$ can be extended to a unitary and  translation-invariant operator on ${\mathcal{H}}_{\hat{\mathcal{C}}}$.
\end{proof}

\end{document}